\documentclass[letterpaper]{article}

\usepackage{fullpage}
\usepackage{amsthm}
\usepackage{url}
\usepackage{epsfig}
\usepackage{graphicx}
\usepackage{amsmath,amssymb}
\usepackage[usenames,dvipsnames]{color}
\usepackage[show]{notes}
\usepackage{algorithmic}
\usepackage[ruled,vlined,linesnumbered,commentsnumbered]{algorithm2e}
\usepackage{multirow}
\usepackage{epstopdf}
\usepackage{graphicx}
\usepackage{type1cm}
\usepackage{eso-pic}
\usepackage{color}
\usepackage{xspace}
\usepackage{hyperref}
\hypersetup{
	colorlinks,breaklinks,
	urlcolor=[rgb]{0,0.25,0.75},
	linkcolor=[rgb]{0,0.25,0.75},
	citecolor=[rgb]{0,0.25,0.75}
}

\newcommand{\eat}[1]{}

\newtheorem{definition}{Definition}

\newtheorem{lemma}{Lemma}
\newtheorem{theorem}{Theorem}

\def\0{{\mathbf 0}}
\def\1{{\mathbf 1}}

\newcommand{\ra}{\mbox{$\rightarrow$}}

\newcommand{\reals}{\mathbb{R}}
\renewcommand{\P}{\mathsf{P}}

\DeclareMathOperator*{\argmax}{argmax}
\newcommand{\SPhard}{\text{\#P}-hard\xspace}
\newcommand{\NPhard}{\text{NP}-hard\xspace}

\newcommand{\MIP}{\mathsf{MIP}}
\newcommand{\MIIA}{\mathsf{MIIA}_\tau}

\newcommand{\InfSet}{InfSet}
\newcommand{\Nin}{\mathcal{N}^{\mathit{in}}}
\newcommand{\EN}{\tilde{\mathcal{N}}^{\mathit{in}}}
\newcommand{\ev}{\mathcal{E}}

\newcommand{\hept}{NetHEPT}
\newcommand{\epi}{Epinions}
\newcommand{\dblp}{DBLP}

\newcommand{\wiki}{WikiVote}

\title{
Time-Critical Influence Maximization in Social Networks with Time-Delayed Diffusion Process\footnote{A preliminary version of this paper appears in {\em Proc.\ of the Twenty-Sixth AAAI Conference on Artificial Intelligence (AAAI 2012)}. The copyright of that conference version is held by the Association for the Advancement of Artificial Intelligence (AAAI).}
}
\author{
\begin{tabular}{ccc}
 	Wei Chen & Wei Lu\footnote{Part of this work was done while Wei Lu was a visiting student at Microsoft Research Asia. Support in part by The Natural Sciences and Engineering Research Council of Canada (NSERC).} \\
	Microsoft Research Asia & University of British Columbia \\
	Beijing, 100080, China & Vancouver, B.C. V6T 1Z4, Canada \\
	\texttt{weic@microsoft.com}\vspace{5pt} & \texttt{welu@cs.ubc.ca}\vspace{5pt} \\
\end{tabular}\\
\begin{tabular}{cc}
	Ning Zhang\footnote{This work was done while Ning Zhang was a visiting student at Microsoft Research Asia.} \\
	Univ. of Science and Technology of China  \\
	Hefei, Anhui, 230026, China \\
	\texttt{lemonustc@gmail.com} \\
\end{tabular}\\
}

\begin{document}

\maketitle

\begin{abstract}
Influence maximization is a problem of finding a small set of
	highly influential users, also known as seeds, in a social network such that the spread of influence
	under certain propagation models is maximized. 
In this paper, we consider time-critical influence maximization,
	in which one wants to maximize influence spread within a given
	deadline.
Since timing is considered in the optimization, we also extend the Independent
	Cascade (IC) model and the Linear Threshold (LT) model to
	incorporate the time delay aspect of
	influence diffusion among individuals in social networks. 
We show that time-critical influence maximization
	under the time-delayed IC and LT models maintains desired properties such as 
	submodularity, which allows a greedy approximation algorithm to 
	achieve an approximation ratio of $1-1/e$.
To overcome the inefficiency of the greedy algorithm, we design two heuristic
	algorithms: the first one is based on
	a dynamic programming procedure that computes exact influence
	in tree structures and directed acyclic subgraphs, while the second one
	converts the problem to one in the original models and then applies
	existing fast heuristic algorithms to it.
Our simulation results demonstrate that our algorithms achieve the
	same level of influence spread as the greedy algorithm while running
	a few orders of magnitude faster, and they also outperform existing
	fast heuristics that disregard the deadline constraint and delays in
	diffusion.
\end{abstract}

\newpage
\section{Introduction}\label{sec:intro}
Recently, the rapidly increasing popularity of online social networking
	sites such as Facebook, Twitter, and Google+ opens up
	great opportunities for large-scale viral marketing campaigns.
Viral marketing, first introduced to the data mining community
	by Domingos and Richardson~\cite{domingos01},
 	is a cost-effective marketing strategy that promotes products
	by giving free or discounted items to a selected group of
	highly influential individuals (seeds), in the hope that through
	the word-of-mouth effects,
	a large number of product adoption will occur.

Motivated by viral marketing, \emph{influence maximization} emerges as a fundamental
	data mining problem concerning the propagation of ideas, opinions, and
	innovations through social networks.
In their seminal paper, Kempe et al.~\cite{kempe03} formulated
	influence maximization as a problem in discrete optimization:
Given a network graph $G$ with pairwise user influence probabilities
	on edges, and a positive number $k$, find $k$ users, such
	that by activating them initially, the expected spread of influence
	is maximized under certain propagation models.
Two classical propagation models studied in the the literature
	are the \emph{Independent Cascade} (IC) and the
	\emph{Linear Threshold} (LT) model. 

The family of models considered in Kempe et al.
	and its follow-ups, including IC and LT, do not fully incorporate important
	{\em temporal} aspects that have been well observed in the dynamics of
	influence diffusion.
First, the propagation of influence from one person to another
	may incur a certain amount of \emph{time delay},
	which is evident from recent studies by 
	statistical physicists on empirical social networks.
Iribarren and Moro~\cite{moro2009}
	observed that the spread of influence in social networks
	slows down due to heterogeneity in human activities.
Karsai et al.~\cite{barabasi2011} reported similar observations
	and attributed such slow-down to the bursty
	nature of human interactions and the topological correlations
	in networks with the small-world property.

Second, the spread of influence may be \emph{time-critical} in practice.
In a certain viral marketing campaign, it might be the case that
	the company wishes to trigger a large volume of product
	adoption in a fairly short time frame, e.g., a three-day sale.
%For example, 
As a motivating example, let us suppose that Alice has bought an
	Xbox 360 console and Kinect with a good discount, but the
	deal would only last for three days.
Alice wanted to recommend this deal to Bob, but whether her recommendation
	would be effective depends on whether Alice and Bob can be in touch (e.g., meeting in person, or Bob seeing the message left by Alice on Facebook) before the discount expires.
%In other words, whether Bob will be influenced by Alice to buy a \textcolor{red}{Zune} with discounts depends on the random event of they meeting within 3 days.
Therefore, when we try to maximize the spread of influence for a viral marketing campaign facing this kind of scenarios, we need to take both the time delay aspect of influence diffusion and the time-critical constraint of the campaign into consideration.
%, which is exactly the goal of this work.

To this end, we extend the influence maximization problem to
	have a \emph{deadline constraint} to reflect the
	time-critical effect.
We also propose two new propagation models, the
	\emph{Independent Cascade model with Meeting events (IC-M)}
	and the {\em Linear Threshold model with Meeting events (LT-M)}
	to capture the delay of propagation in time.
We show that both the IC-M and LT-M models maintain desired properties, namely monotonicity and submodularity, which implies a greedy $(1-1/e)$-approximation algorithm in spite of the \NPhard{ness} of the problem.

For the IC-M model, we design two efficient and effective heuristic algorithms, MIA-M and MIA-C, based on the notion of Maximum Influence Arborescence (MIA)~\cite{ChenWW10} to tackle time-critical influence maximization.
On the other hand, similarly for the LT-M model, we adapt the notion Local Directed Acyclic Graph (LDAG) proposed in~\cite{ChenYZ10}
	to obtain the LDAG-M heuristic algorithm.
Our experiments evaluate performance of various algorithms for the IC-M model, including the greedy approximation algorithm, MIA-M, and MIA-C heuristics.
The empirical results demonstrate that our heuristic algorithms produce seed sets with equally good quality as those mined by approximation algorithm, while being two to three orders of magnitude faster.
Moreover, we show that only using standard heuristics such as MIA and 
	disregarding time delays and deadline constraint 
	could result in poor influence spread compared to
	our heuristics that are specifically designed
	for this context.

\subsection{Related Work}\label{sec:related}
Domingos and Richardson~\cite{domingos01,richardson02}
	first posed influence maximization as an algorithmic problem.
They modeled the problem using Markov random fields and proposed
	heuristic solutions.
Kempe et al.~\cite{kempe03} studied influence maximization
	as a discrete optimization problem.
They showed that the problem is \NPhard{} under both
	the IC and LT models, and relied on submodularity
	to obtain a $(1-1/e)$ greedy approximation
	algorithm.

A number of studies following \cite{kempe03} developed more efficient
	and scalable solutions, including the
	cost-effective lazy forward (CELF) optimization~\cite{Leskovec07}
%	that further exploits submodularity to gain significant speed-up for
%	the greedy algorithm,
	and also work by Kimura et al.~\cite{KimuraSN07},
	Chen et al.~\cite{ChenWY09}, etc.
Specifically for the IC model, Chen et al.~\cite{ChenWW10}
	showed that it is \SPhard to compute the exact influence of any node
	set in general graphs.
They proposed the MIA model which uses influence in local tree structures
	to approximate influence propagated through the entire network.
They then developed scalable algorithms to compute exact
	influence in trees and mine seed sets with equally good quality
	as those found by the approximation algorithm.
For the LT model, it is also \SPhard to compute the exact influence in general graphs~\cite{ChenYZ10}.
To circumvent that,
	Chen et al.~\cite{ChenYZ10} proposed to use
	local directed acyclic graphs (LDAG) to
	approximate the influence regions of nodes.

More recently, Chen et al.~\cite{ChenNegOpi11} extended the IC model
	to capture the propagation of negative opinions in information diffusion.
Goyal, Bonchi, and Lakshmanan~\cite{GoyalBL11} leveraged real propagation
	traces to derive more accurate diffusion models.
In~\cite{mintime}, Goyal et al. studied the problem of MINTIME in influence
	maximization.
In MINTIME, an influence spread threshold $\eta$ and a budget threshold $k$ are given,
	and the task is to find a seed set of size at most $k$ such that by activating it,
	at least $\eta$ nodes are activated in expectation in the minimum possible time.
This problem also considers timing in influence maximization, but the number of
	time steps in MINTIME is the optimization objective, whereas in our case,
	time is given as a deadline constraint.

The time-delay phenomena in information diffusion
	has been explored in statistical physics.
Iribarren and Moro~\cite{moro2009} observed from a large-scale
	Internet viral marketing experiment in Europe that the dynamics of
	information diffusion are controlled by the heterogeneity
	of human activities.
More recently, using time-stamped phone call records, 
	Karsai et al.~\cite{barabasi2011} found that the spreading
	speed of information on social networks is much slower than one
	may expect, due to various kinds of correlations, such as
	community structures in the graph, weight-topology correlations, and
	bursty event on single edges.

\section{Time-Critical Influence Maximization for Time-Delayed Independent Cascade Model}\label{sec:model}
\subsection{Model and Problem Definition}\label{sec:probDef}

We first describe the standard Independent Cascade (IC) model in Kempe et al.~\cite{kempe03}, 
	and then show our extensions that incorporate deadline and random meeting events.
In the IC model, a social network is modeled as a directed graph $G=(V,E)$,
	where $V$ is the set of nodes representing %individual 
	users and $E$ is the set of directed
	edges representing links (relationship) between users.
Each edge $(u,v)\in E$ is associated with an influence probability $p(u,v)$ 
	defined by function $p:E\rightarrow [0,1]$.
If $(u,v)\not\in E$, define $p(u,v)=0$.

The diffusion process under the IC model proceeds in discrete
	time steps $0,1,2,\dotso$.
Initially, a {\em seed set} $S\subseteq V$ is targeted and activated 
	at step $0$, while all other nodes are inactive.
At any step $t\ge 1$, any node $u$ activated at step $t-1$ is given a 
	single chance to activate any of its currently inactivate neighbors $v$
	with independent success probability $p(u,v)$.
Once a node is activated, it stays active.
The process continues until no new nodes can be activated.
The influence maximization problem under the IC model is to find a seed
	set $S$ with at most $k$ nodes such that the expected number of
	activated nodes after the diffusion terminates, 
	called {\em influence spread} and denoted by $\sigma(S)$, is maximized.

%% In practice, one may only care about the influence spread within a certain
%%     time period, thus we would like to extend the influence maximization
%% 	problem to have a deadline constraint.
%% Before defining the new problem, we first extend the IC model to incorporate
%% 	a more realistic time delay aspect of activation, that is, the activation
%% 	attempt of $u$ to its neighbor $v$ may not occur immediately at the
%% 	next step after $u$ itself is activated, but may incur a time delay.

We now describe our extension to the IC model to incorporate time-delayed
	influence diffusion, which we denote by IC-M (for Independent Cascade with
	Meeting events). %as follows.
In the IC-M model, each edge $(u,v)\in E$ is also associated with a \emph{meeting probability} $m(u,v)$
	defined by function $m:E\rightarrow [0,1]$
	(if $(u,v) \notin E$, $m(u,v)=0$).
As in IC, a seed set $S$ is targeted and activated at step $0$.
At any step $t \ge 1$, an active node $u$ meets any of 
	its currently inactive neighbors $v$ independently with probability $m(u,v)$.
If a meeting event occurs between $u$ and $v$ for the \emph{first} time,
	%after $u$ is activated, 
	$u$ is given a \emph{single} chance to try activating $v$,
	with an independent success probability $p(u,v)$.
If the attempt succeeds, $v$ becomes active at step $t$ and will start 
	propagating influence at $t+1$.
%Otherwise, $u$ will not be able to activate $v$ even if they would  
%	meet again in future steps.
The diffusion process quiesces when all active nodes have met with all their neighbors
	and no new nodes can be activated.

Several possibilities can be considered in mapping the meeting events
	in the IC-M model to real actions in online social networks.
For instance, a user $u$ on Facebook posting a message on her
	friend $v$'s wall can be considered as a meeting event.
In the IC-M model, the meeting probabilities are not necessarily the
	same for different pair of users.
Different pairs of friends may have different frequencies of exchanging
	messages on each other's walls, which is reflected by the
	meeting probability. 

Note that the original IC model is a special case of IC-M with 
	$m(u,v)=1$ for all edges $(u,v)\in E$.
More importantly, for the original influence maximization problem, the meeting
	probability is not essential, because as long as $m(u,v)>0$, eventually
	$u$ will meet with $v$ and try to influence $v$ once.
Thus, if we only consider the overall influence in the entire run, there
	would be no need to introduce meeting probabilities.
However, if we consider influence within a \emph{deadline constraint}, then
	meeting probability is an important factor in determining the optimal
	seed set.

Formally, for a deadline $\tau\in \mathbb{Z}_+$, we define $\sigma_\tau\colon 2^V\to \reals_+$
	to be the set function such that $\sigma_\tau(S)$ with $S\subseteq V$
	is the expected number of activated nodes by the end of time step $\tau$ under
	the IC-M model, with $S$ as the seed set.
The {\em time-critical influence maximization with a deadline constraint $\tau$} 
	is the problem of finding the seed set $S$ with at most $k$ seeds such that
	the expected number of activated nodes by step $\tau$ is maximized, 
	i.e., finding $S^* = \arg\max_{S\subseteq V, |S|\le k} \sigma_\tau(S)$.

Note that the original influence maximization problem is
	\NPhard{} for the IC model~\cite{kempe03}
	and that problem is a special case of time-critical influence
	maximization for the IC-M model with all $m(u,v) = 1$
	and deadline constraint $\tau=|V|$.
This leads to the following hardness result.

\begin{theorem}\label{theorem:hard}
The time-critical influence maximization problem is \NPhard{}
	for the IC-M model.
\end{theorem}

\subsection{Properties of the IC-M Model}
Although to find the optimal solution for time-critical influence maximization with deadline $\tau$ under IC-M is \NPhard{} (Theorem~\ref{theorem:hard}), we show that the influence function $\sigma_\tau(\cdot)$ is monotone and submodular, which allows a hill-climbing-style greedy algorithm to achieve a $(1-1/e)$-approximation to the optimal.

Given a ground set $U$, a set function $f\colon2^U\to \reals$ is \textsl{monotone} if $f(S_1) \leq f(S_2)$ whenever $S_1 \subseteq S_2$. 
Also, the function is \textsl{submodular} if $f(S_1 \cup \{w\}) - f(S_1) \geq f(S_2 \cup \{w\}) - f(S_2)$, $\forall S_1 \subseteq S_2$, $\forall w \in U \setminus S_2$.
Submodularity captures the law of diminishing marginal returns, a well-known principle in economics. 

\begin{theorem}\label{thm:submod}
	The influence function $\sigma_\tau(\cdot)$ is monotone and 
	submodular for an arbitrary instance of the IC-M model, 
	given any deadline constraint $\tau \ge 1$.
\end{theorem}

To prove the theorem, we can view the random cascade process
	under IC-M using the ``possible world'' semantics and the
	principle of deferred decisions.
%, similar to IC is viewed in Kempe et al.~\cite{kempe03}. 
That is, we can suppose that before the cascade starts, a set of outcomes for all meeting events, as well as the ``live-or-blocked'' identity for all edges are already determined but not yet revealed.

More specifically, for each meeting event (a $(u, v)$ pair and a time step $t \in [1,\tau]$), we flip a coin with bias $m(u,v)$ to determine if $u$ will meet $v$ at $t$.
Similarly, for each edge $(u,v) \in E$, we flip once with bias $p(u,v)$, and we declare the edge ``live'' with probability $p(u,v)$, or ``blocked'' with probability $1-p(u,v)$.
All coin-flips are independent.
The identity of the edge $(u,v)$ is revealed \emph{in the event} that $u$ is active and is meeting the inactive $v$ for the first time.
Therefore, a certain set of outcomes of all coin flips corresponds to one {\em possible world}, denoted by $X$, which is a deterministic graph (with all blocked edges removed) obtained by conditioning on that particular set of outcomes.

Now we prove Theorem~\ref{thm:submod}.

\begin{proof}[Proof of Theorem~\ref{thm:submod}]
	%When $\tau = 1$, $\sigma_\tau(S) \equiv |S|$ ($\forall S \subseteq V$), which is trivially monotone submodular.
	%When $\tau > 1$, 
Fix a set $X_M$ of outcomes of all meeting events ($\forall (u,v)\in E$, $\forall t\in [0,\tau]$), and also a set $X_E$ of live-or-blocked identities for all edges.
Since the coin-flips for meeting events and those for live-edge selections are orthogonal, and all flips are independent, any $X_E$ on top of an $X_M$ leads to a possible world $X$.

Next, we define the notion of ``reachability'' in $X$.
%which is different from its traditional definition due to the deadline constraint and meeting events. 	
Consider a live edge $(u,v)$ in $X$.
Traditionally, without meeting events, $v$ is reachable from $u$ via just one hop.
	%Now, with pre-determined meeting sequences, $v$ is reachable from $u$ via $h_{u,v}$ hops, where $h_{u,v}$ is the difference between the first step in which $u$ meets $v$ after $u$ is reached, and the step in which $u$ itself is reached.
	Now with pre-determined meeting sequences, $v$ is reachable from $u$ via $t_v - t_u$ hops, where $t_u$ is the step in which $u$ itself is reached, and $t_v$ is the first step when $u$ meets $v$, after $t_u$.
%the number of steps from the step $u$ is reached (if $u$ is a seed then it is ``reached'' at step $0$) to the step the nearest next meeting occurs for $u$ and $v$.
	Hence, we say that $v$ is \emph{reachable} from a seed set $S$ if and only if (1) there exists at least one path consisting entirely of live edges (called live-path) from some node in $S$ to $v$, and (2) the \emph{collective number of hops} along the \emph{shortest} live-path from $S$ to $v$
%, determined by the sequences of meeting time of all $(u_i, u_{i+1}) \in \P$, 
is no greater than $\tau$.

Then, let $\sigma^X_\tau(S)$ be the number of nodes reachable from $S$ in $X$ (by the reachability definition above).
Let $S_1$ and $S_2$ be two arbitrary sets such that $S_1\subseteq S_2\subseteq V$, and let node $w \in V \setminus S_2$ be arbitrary.
The monotonicity of $\sigma^X_\tau(\cdot)$ holds, since if some node $u$ can be reached by $S_1$, the source of the live-path to $u$ must be also in $S_2$.
%; this gives $\sigma^X_\tau(S_1) \leq \sigma^X_\tau(S_2)$.
As for submodularity, consider a certain node $u$ which is reachable from $S_2 \cup \{w\}$ but not from $S_2$.
This implies (1) $u$ is not reachable from $S_1$ either, and (2) the source of the live-path to $u$ must be $w$.
Hence, $u$ is reachable from $S_1 \cup \{w\}$ but not from $S_1$.
This gives $\sigma^X_\tau(S_1\cup\{w\})-\sigma^X_\tau(S_1) \geq \sigma^X_\tau(S_2\cup\{w\})-\sigma^X_\tau(S_2)$.

Let $\mathcal{E}_I$ denote the event that $I$ is the true realization (virtually) of the corresponding random process.
Taking the expectation over all possible worlds, we have $\sigma_\tau(S) = \sum_{X} \Pr[\mathcal{E}_X] \cdot \sigma^X_\tau(S), \forall S\subseteq V$, where $X$ is any combination of $X_E$ and $X_M$, and $\Pr[\mathcal{E}_X] = \Pr[\mathcal{E}_{X_E}]\cdot\Pr[\mathcal{E}_{X_M}]$.
Therefore, $\sigma(\cdot)$ is a nonnegative linear combination of monotone and submodular functions, which is also monotone and submodular.
This was to be shown.
%A nonnegative linear combination of submodular functions is also submodular, and hence $\sigma_\tau(\cdot)$ is monotone and submodular.
\end{proof}

\section{Time-Critical Influence Maximization under Time-Delayed Linear Threshold Model}\label{sec:ltmodel}
\subsection{Model and Problem Definition}

In the LT model, a social network is modeled as a directed
	graph $G=(V,E)$, where $V$ is the set of nodes
	representing individuals and $E$ is the set of directed
	edges representing links (relationships, ties, etc.) between
	individuals.
Each edge $(u,v)\in E$ is associated with an influence weight
	$b(u,v)$ defined by function $b\colon E\to [0,1]$.
We require that $\sum_{u}b(u,v) \le 1$, $\forall v\in V$.
If $(u,v)\not\in E$, define $b(u,v) = 0$.
Each node $v$ chooses a {\em threshold} $\theta_v$
	uniformly at random from the interval $[0,1]$,
	which represents the weighted fraction of $v$'s
	in-neighbors that must be active so as to let
	$v$ become active.

Given the random choices of node thresholds, and a
	seed set $S\subseteq V$, the dynamics of the
	diffusion process under the LT model
	proceed in discrete time steps $0,1,2,\dotso$.
Initially, at step $0$, nodes in $S$ are activated.
Then at any step $t \ge 1$, nodes who have been
	active in previous steps remain active, and
	we activate any node $v$ such that its threshold
	$\theta_v$ is surpassed by the total weights
	of its currently active in-neighbors.
That is, $\sum_{\text{active }\;u\in \Nin(v)} b(u,v) \ge \theta_v$.
The process continues until no new nodes
	can be activated.
The influence maximization problem under the LT
	model is to find a seed set $S$ with $|S|\le k$
	such that the expected number of activated
	nodes after the diffusion process terminates,
	called {\em influence spread} and denoted by $\sigma(S)$,
	is maximized.

We now describe our extension to the LT model
	that incorporates time-delayed diffusion processes,
	which we call LT-M (Linear Threshold model with
	Meeting events).
In the LT-M model, each edge $(u,v)\in E$ is also
	associated with a {\em meeting probability}
	$m(u,v)$ defined by function $m\colon E\to [0,1]$.
Note that $m(u,v) = 0$ if $(u,v)\not\in E$.
Same as in LT, nodes choose a uniform random number
	out of $[0,1]$ as threshold, and a seed set
	$S$ is targeted and activated at step $0$.

At any step $t\ge 1$, each active node $u$ remains active and
	it meets any of its currently inactive out-neighbors $v$
	with probability $m(u,v)$, independently.
If since $u$'s activation, $u$ and $v$ meet for the first
	time at $t$, then we say $u$'s influence weight to $v$
	is effective.
An inactive $v$ becomes active if the total
	effective weight from its active in-neighbors is at least $\theta_v$.
That is,
$$
\sum_{u\in \EN(v)} b(u,v) \ge \theta_v,
$$
	where $\EN(v) \subseteq \Nin(v)$ denotes the set
	of effective in-neighbor of $v$.
After $v$ becomes active, it starts to meet her neighbors
	from the next time step.
The process quiesces when all active nodes have met
	with their out-neighbors and no new nodes can
	be activated.

Note that the original influence maximization problem is \NPhard\
	for the LT model~\cite{kempe03} and that is a special
	case of time-critical influence
	maximization under LT-M with all $m(u,v)$ = 1 and
	deadline $\tau = |V|$.
This leads us to the following hardness result.

\begin{theorem}\label{thm:lthard}
The time-critical influence maximization problem
	is \NPhard\ for the LT-M model.
\end{theorem}

\subsection{Properties of the LT-M Model}

\begin{theorem}\label{thm:ltsubmod}
The influence function $\sigma_\tau(\cdot)$ is monotone and 
	submodular for an arbitrary instance of the LT-M model, 
	given any deadline constraint $\tau \ge 1$.
\end{theorem}

To prove Theorem~\ref{thm:ltsubmod}, let us first give
	an alternative form of the model definition for LT-M.
The main idea is to show that the diffusion process guided
	by LT-M is equivalent to one guided by a random
	``live-edge'' selection process.
Recall that in LT-M, each edge $(u,v)$ has an weight
	$b(u,v)$, and $\sum_{u}b(u,v)\le 1$ for all $v$.
The random edge selection protocol lets $v$ pick at
	most one incoming edge (and the corresponding in-neighbor) independently at random,
	selecting a particular $u$ with probability $b(u,v)$,
	or selecting no one with probability $1-\sum_u b(u,v)$.
If $v$ selects $u$, then we declare edge $(u,v)$ ``live'';
	otherwise declare it ``blocked''.
Kempe, Kleinberg and Tardos~\cite{kempe03}
	showed that the distribution over the final active
	node sets obtained by running the Linear Threshold
	process is equivalent to that by running the above
	random live-edge selection process.

In our case, it is also necessary to incorporate the random
	meeting events in LT-M into the ``live-edge'' (LE)
	process and obtain the ``LE-M'' model.
Specifically, now, at any step $t$, we activate any inactive
	$v$ provided that its selected neighbor $u$ was activated
	at some earlier step $t' < t$, and since $t'+1$,
	they meet for the first time here at $t$.
\begin{lemma}\label{lemma:le}
The Linear Threshold model incorporated with random meeting events is
	equivalent to the live-edge model incorporated with
	random meeting events.
In other worlds, the distribution over the
	final active sets under the LT-M model is the same as
	the one under the LE-M model.
\end{lemma}

\begin{proof}
Let $A_t(v)$ be the set of nodes that are already active and
	have met $v$ at least once by the end of step $t$.
%Notice that $A_0$ is just our seed set $S$.
For the linear threshold process,
if some $v$ is not yet active by the end of step $t$,
	the chance that $v$ will become active
	at the next step, $t+1$, is the probability that the
	incremental weight contributed by
	$A_{t}(v)\setminus A_{t-1}(v)$
	manages to push the total effective weight as of $t+1$ over $\theta_v$.
Denote this probability by $P_1(v,t+1)$.

For the live-edge process, we analyze the probability of the
	same event, that is, an inactive $v$ will
	be activated at $t+1$, given that it has not been
	active yet by the end of $t$.
Let $\ev_1$ be the event that the corresponding node $u$ of
	the live-edge $v$ selected becomes active at $t$,
	and $u$ meets $v$ right away at $t-1$.
Let $\ev_2$ be the event that this selected $u$
	has become active before $t$, but $u$ and $v$
	have not met until $t+1$.
The probability that we want equals to $\Pr[\ev_1\cup\ev_2]$.
Notice that since $v$ selects at most one live-edge,
	$\ev_1$ and $\ev_2$ are mutually exclusive.
Let this probability be $P_2(v,t+1)$.

What remains to be shown is that $P_1(v,t+1) = P_2(v,t+1)$.
Notice that the meeting events are all independent
	and orthogonal to the live-edge selection
	or the linear threshold process.
Thus, by the principle of deferred decision, we can treat
	these meeting events as if their outcomes were determined before
	the diffusion starts, but were only going to get revealed
	as the process proceeds.
Without loss of generality, we fixed a set $M$ of outcomes,
	by independently flipping a coin with bias $m(u,v)$,
	$\forall (u,v)\in E$ in all time steps.

Conditioning on the fixed $M$, for any $v\in V$ and $t\in\{1,\dotsc,\tau-1\}$, we have
$$
P^M_1(v,t+1) = P^M_2(v,t+1) = \frac{\sum_{u\in (A_t(v)\setminus A_{t-1}(v))} b(u,v) }{1-\sum_{u\in A_{t-1}(v)} b(u,v) }\,.
$$
Then, by un-conditioning, i.e., taking the expectation over
	all possible outcomes of meeting events, we have $P_1(v,t+1) = P_2(v,t+1)$.
This completes the proof.

\end{proof}

With Lemma~\ref{lemma:le}, we can apply the same argument
	used in the submodularity proof for the IC-M model,
	to prove Theorem~\ref{thm:ltsubmod}.
The flow of the proof and the techniques used in it
	is similar to the proof of Theorem~\ref{thm:submod}.
We first run the live-edge selection random process and
	flip coins for all meeting events.
Conditioning on a fixed set of meeting events outcomes, and
	removing all edges but the live ones, we obtain a
	deterministic graph (possible world).
Then following the same arguments used in the proof of
	Theorem~\ref{thm:submod} leads us to the submodularity
	of the influence function $\sigma_\tau(\cdot)$ for the
	LT-M model.

\paragraph{Approximation Guarantees.}
Thus far, we have shown that the influence function 
	$\sigma_\tau(\cdot)$
	is monotone and submodular under both IC-M and LT-M models,
	then our time-critical
	influence maximization problem is a special case
	of monotone submodular function maximization
	subject to a cardinality (uniform matroid) constraint.
Therefore, we can apply the celebrated result
	in Nemhauser et al.~\cite{submodular} to
	obtain a greedy $(1-1/e)$-approximation algorithm.
The greedy algorithm repeatedly grows $S$ by adding $u$
	with the largest marginal influence w.r.t $S$ in each iteration
	until $|S|=k$.

\paragraph{Generalizing to the Triggering Set Model.}
So far we have obtained a constant factor approximation algorithm
	for the time-critical influence maximization under both IC-M
	and LT-M models (Theorem~\ref{thm:submod} and Theorem~\ref{thm:ltsubmod}).
In fact, this result can be applied to a more general time-delayed
	propagation model called the Triggering Set model
	with Meeting events (TS-M).

In the original Triggering Set (TS) model~\cite{kempe03},
	each node $v\in V$ independently chooses a random
	triggering set $\mathcal{T}(v)\subseteq \Nin(v)$, according
	to some distribution over all subsets of $\Nin(v)$.
During the influence diffusion process guided by the TS
	model, a node $v$ becomes active at step $t$
	if there is some $u\in \mathcal{T}(v)$ that is
	active at step $t-1$.
Viewing the process using the live-edge model, we declare the
	edge $(u,v)$ live if $u$ is selected into
	the triggering set of $v$; otherwise it is declared blocked.

The Triggering Set model generalizes IC and LT for the following reasons.
Recall that for IC, each edge $(u,v)$ in the graph is live with
	an independent probability $p(u,v)$.
Thus, for any $v$, we can view the triggering set selection
	process as $v$ adds each in-neighbor $u$ independently.
Meanwhile, for LT, recall that it is equivalent to the live-edge
	selection process, which amounts to saying that $v$ picks
	at most one in-neighbor $u$ into its triggering set with probability $b(u,v)$, and
	picks an empty set with probability $1-\sum_u b(u,v)$.

The meeting events are orthogonal and independent
	from the triggering set selection process, and thus
	the TS-M model still generalizes the IC-M and LT-M models.
Using the same arguments in the proofs for Theorem~\ref{thm:submod}
	and Theorem~\ref{thm:ltsubmod}, the influence function
	$\sigma_\tau(\cdot)$
	is monotone and submodular for every instance of the TS-M
	model as well, and thus this general model also enjoys
	the approximation guarantees provided by the greedy algorithm.

\paragraph{Inefficiency of the Greedy Approximation Algorithm.}
Although the greedy approximation algorithm, it is \SPhard to compute the exact influence in general
	graphs for the IC and LT models~\cite{ChenWW10, ChenYZ10}.
The hardness applies to IC-M and LT-M, since
	each of them subsumes corresponding original model.
A common practice is to estimate influence spread using Monte-Carlo (MC) simulations, in which case the approximation ratio of Greedy drops to $1-1/e-\epsilon$, where $\epsilon$ is small if the number of simulations is sufficiently large.
Due to expensive simulations, the greedy algorithm is not scalable to large data, even the implementation can be accelerated by the CELF optimization~\cite{Leskovec07}.

\begin{algorithm}[t!]
\caption{Greedy ($G=(V,E)$, $k$, $\sigma_\tau$)}\label{alg:greedy}
$S \gets \emptyset$\;
\For {$i = 1 \rightarrow k$} {
	$u \gets \argmax_{v\in V\setminus S}\big[\sigma_\tau(S\cup\{v\})-\sigma_\tau(S)\big]$\;
	$S \gets S \cup \{u\}$\;
}
Output $S$\;
\end{algorithm}

\section{Computing Influence in Arborescences in the IC-M model}\label{sec:tree}
%Since the simple greedy algorithm lacks of a way to computing influence, it is slow and not practical.
In this section, we derive an dynamic programming algorithm that computes exact influence spread in tree structures, which will be used in Sec.~\ref{sec:algo} to develop MIA-M.
The algorithmic problem of efficient computation of exact influence spread
	in trees in the IC-M model with deadline constraint is also of independent
 	interest.

An \emph{in-arborescence} is a directed tree where all edges point into the root.
Given a graph $G=(V,E)$ with influence probability function $p$ and meeting probability function $m$, consider an in-arborescence $A=(V_A,E_A)$ rooted at $v$ where $V_A\subseteq V$ and $E_A \subseteq E$.
We assume that influence propagates to $v$ only from nodes in $A$.
We also assume that there exists at least one $s\in S$ such that $s \in V_A$; otherwise no nodes can be activated in $A$.
Given a seed set $S$ and deadline $\tau$, we show how to compute $\sigma_\tau(S)$ in $A$ in time linear to the size of the graph.
%Fix a seed set $S \subseteq V$ and a deadline $\tau$, we study algorithms that compute $\sigma_\tau(S)$ in $A$.
%Since $S$, $\tau$, and $A$ are fixed in this section, we omit them in notations below.

Let $ap(u,t)$ be the {\em activation probability} of $u$ at step $t$, i.e., the probability that $u$ is activated at step $t$ after the cascade ends in $A$.
Since the events that $u$ gets activated at different steps are mutually exclusive, the probability that $u$ ever becomes active by the end of step $\tau$ is $\sum_{t=0}^\tau ap(u,t)$.
By linearity of expectation, $\sigma_\tau(S) = \sum_{u \in V}\sum_{t=0}^\tau ap(u,t)$.
Hence, the focus is to compute $ap(u,t)$, for which we have the following theorem.
%and we have the following theorem:
\begin{theorem}\label{thm:ap}
Given any $u$ in arborescence $A$, and any $t \in [0,\tau]$,
	the {\em activation probability} $ap(u,t)$ can be recursively
	computed as follows.

For base cases when $u\in S$ or $t=0$,
\begin{equation}
ap(u,t) =
	\left\{
	\begin{array}{lr}
 	1 &\mbox{ $(u \in S\wedge t = 0)$}\\
 	0 &\mbox{ $(u \not\in S\wedge t = 0)$}\\	
 	0 &\mbox{ $(u \in S\wedge t \in \{1,\dotsc,\tau\})$}
	\end{array}
	\right.
\end{equation}

For $u \not\in S\wedge t \in \{1,\dotsc,\tau\}$,
\begin{align}\label{eqn:ap}
ap(u,t)
\nonumber &= \prod_{w \in \Nin(u)} \Big( 1-\sum_{t'=0}^{t-2} ap(w,t')\,p(w,u)\,[1-(1-m(w,u))^{t-t'-1}]\Big)\\
		&\qquad - \prod_{w \in \Nin(u)} \Big( 1-\sum_{t'=0}^{t-1} ap(w,t')\,p(w,u)\,[1-(1-m(w,u))^{t-t'}]\Big),
\end{align}
where $\Nin(u) \subseteq V_A$ is the set of in-neighbors of $u$ in $A$.
\end{theorem}

\begin{proof}
	The base cases ($u\in S$ or $t=0$) are trivial.
	When $u \not\in S$ and $t \in \{1,\dotsc,\tau\}$, for any in-neighbor $w \in V_A$ of $u$ and $t' < t$, $p(w,u)\big(1-(1-m(w,u))^{t-t'-1}\big)$ is the probability that $w$ meets $u$ at least once from $t'+1$ to $t-1$ and that $(w,u)$ is live.
	Since the events that $w$ gets activated at different $t'$ are mutually exclusive, $1-\sum_{t'=0}^{t-2}ap(w,t')p(w,u)\big(1-(1-m(w,u))^{t-t'-1}\big)$ is the probability that $u$ has not been activated by $w$ before or at $t-1$.
	Note that $\sum_{t'=0}^{-1} ap(w,t')p(w,u)\big(1-(1-m(w,u))^{t-t'-1}\big)=0$, so the above still holds for $t=1$.	
	Similarly, $\prod_{w \in N^{\text{in}}(u)}\big( 1-\sum_{t'=0}^{t-1} ap(w,t')p(w,u)\big(1-(1-m(w,u))^{t-t'}\big)\big)$ is the probability that $u$ has not become active before or at $t$.
	Hence, Formula~\ref{eqn:ap} is exactly the probability that $u$ is activated at $t$, which is $ap(u,t)$.
\end{proof} 

%\textbf{
The recursion given by Formula~\ref{eqn:ap} can be carried out by dynamic programming, traversing from leaves to the root.
%In a general in-arborescence, given node $u$ and a step $t$ as input, calculating $ap(u,t)$ by Formula~\ref{eqn:ap} requires $O(t)$ time, exponential w.r.t the size of the input ($\log t$ bits).
In a general in-arborescence, given as input a node $u$ and a deadline constraint $\tau$, the time complexity of calculating $\sum_{t=0}^\tau ap(u,t)$ by Formula~\ref{eqn:ap} is polynomial to $\tau$, which is exponential to the size of the input: $\Theta(\log\tau)$ bits.
In principle, this does not affect efficiency much as $\tau$ is small ($5$ or $10$), and in general much smaller than the size of the graph.
%}

To reduce the amount of computations, a few optimizations can be applied in implementation.
Let $path(u)$ be the path from some $s \in S$ in $A$ to $u$ that has the minimum length among all such paths.
Note that we only need to compute $ap(u,t)$ for $t \in \{|path(u)|,\dotsc,\tau\}$, as $u$ cannot be reached earlier than step $|path(u)|$.
That is, $ap(u,t)=0$ when $t < |path(u)|$.
Also, if $path(u) = \emptyset$ (i.e., does not exist), $ap(u,t)=0, \forall t$.

For computing $ap(u,t)$ on a chain of nodes within an in-arborescence, we derive a more efficient method that reduces the computation to polynomial to $\log \tau$, as shown in the next section.
%We defer the discussions on this to Section~\ref{sec:discuss}.

\subsection{Fast Influence Computation on Chain Graphs}\label{sec:chain}

Let $r\in \{1,\dotsc,|V|\}$.
We consider a length-$r$ (directed) chain graph $L = (V_L, E_L)$, where
	$V_L=\{u_0,u_1, \ldots, u_r\}$
	and $E_L=\{(u_i,u_{i+1}) : i \in \{0,\dotsc r-1\}\}$:
$$
u_0 \longrightarrow u_1 \longrightarrow \ldots \longrightarrow u_{r-1} \longrightarrow u_r \,.
$$
Here, $L$ can be thought of as a ``sub-arborescense'' of an in-arborescence $A$.
The problem is to compute $ap(u, t)$, given a step $t \in \{0,\dotsc,\tau\}$,
	for any node $u$ on this chain.
Notice that for any particular node $u$,
	what really matters in our analysis is the distance, i.e., the length of the path
	from the closest effective seed to $u$.
Suppose that there is another seed $u_s \neq u_0$
	on the chain, where $s < r$.
Nodes $u_1, \dotsc, u_{s-1}$ will only be possibly influenced by
	$u_0$ (which is the effective seed for them), due to the
	structure of the chain,
	while nodes $u_{s+1},\dotsc,u_r$ will only be possibly
	influenced by $u_s$ (the effective seed), since the influence
	from $u_0$ and any node before $u_s$
	will be blocked by $u_s$.
Therefore, without loss of generality, we assume that the seed set
	$S$ is a singleton set $\{u_0\}$ and the sole seed $u_0$ is activated
	at step $0$, i.e., $ap(u_0,0) = 1$ and $ap(u,t) = 0$ for any $t \ge 1$. 

Denote by $m_i \overset{\text{def}}{=} m(u_{i-1}, u_i)$ the meeting probability
	and by $p_i \overset{\text{def}}{=} p(u_{i-1}, u_i)$ the influence probability
	between $u_{i-1}$ and $u_i$, respectively, for all $ i \in \{1,\dotsc, r\}$.
The activation of node $u$ at step $t$
	can be characterized by a probability distribution over all possible $t$,
	taking all relevant meeting probabilities and influence probabilities
	into consideration.
%% More specifically, let $Y_u$ be the random variable associated
%% 	with node $u$, indicating the time step at which $u$ would
%% 	be activated, given seed set $\{u_0\}$.
%% Clearly, $\Pr[Y_u = t] = ap(u,t)$, and we are interested
%% 	in the distribution (probability mass function) of $Y_u$.

For all $ i \in \{1,\dotsc, r\}$, let $X_i$ be the random variable indicating
	the number of steps needed for node $u_{i-1}$ to meet $u_i$ for the
	first time, given meeting probability $m_i$.
Note that $X_i$ is a geometric random variable with parameter $m_i$, and taking
	values from $\{1,2,3,\ldots\}$.
Suppose that the path from $u_0$ to $u$ is of length $\ell$.
Let $X = \sum_{i=1}^{\ell} X_i$.
The following lemma links the activation probability $ap(u,t)$ with
	the probability distribution of $X$.

\begin{lemma} \label{lem:aput}
The activation probability of node $u$ at time $t$ is given as
\[
ap(u,t) =\left( \prod_{i=1}^\ell p_i \right) \cdot \Pr[X=t].
\]

\end{lemma}
\begin{proof}
The event that $u$ is activated at time $t$ is equivalent to that
	for all $i\in\{1,\dotsc, \ell\}$, 
	node $u_{i-1}$ meets $u_i$ and activates $u_i$ upon their first meeting, and
	the first meeting of $u_{\ell-1}$ and $u_\ell$ after $u_{\ell-1}$
	is activated occurs at time $t$.
Since $X= \sum_{i=1}^{\ell} X_i$, $\Pr[X=t]$ is the probability that
	$u_0$ meets $u_1$, and then $u_1$ meets $u_2$, and so on, and
	$u_{\ell-1}$ meets $u_\ell$, and the total number of steps taken until
	$u_{\ell-1}$ meets $u_\ell$ is $t$.
$\prod_{i=1}^\ell p_i$ is the probability that all activation attempts
	are successful upon first meetings, conditioned on the event that meetings 
	of $u_{i-1}$ and $u_i$ for all $i\in\{1,\dotsc, \ell\}$ occur.
Therefore, it is clear that $\left( \prod_{i=1}^\ell p_i \right) \cdot \Pr[X=t]$ is the
	probability that $u=u_\ell$ is activated at time $t$.
\end{proof}

With Lemma~\ref{lem:aput}, the key to compute activation probability
	$ap(u,t)$ is to compute the probability distribution of $X$, the
	sum of $\ell$ geometric random variables.
When all geometric random variables have the {\em same} parameter, $X$
	is the well-known negative binomial random variable.
We restate the result on the distribution of negative binomial random variable
	below.
For completeness, a proof is given in the appendix.

\def\lemnegbinomial{
Let $X=\sum_{i=1}^{\ell} X_i$, where each $X_i$ is a geometric random variable
	with parameter $m$ and range $\{1,2,3,\ldots\}$.
Then we have $\Pr[X=t]=0$ for $t < \ell$, and for all $t\ge \ell$,
$$
\Pr[X=t] = m^\ell \cdot (1-m)^{t-\ell} \cdot \binom {t-1}{\ell-1}.
$$
}
\begin{lemma}[Negative Binomial Distribution] \label{lem:negbinomial}
{\lemnegbinomial}
\end{lemma}

For the summation of geometric random variables with different parameters, 
	the result is more complex.
We cannot find the result from literature for the general case, 
	so we provide our
	own analysis for the case that all parameters of geometric random 
	variables are {\em distinct}, as shown below.

\def\lemalldistinct{
Let $X=\sum_{i=1}^{\ell} X_i$, where each $X_i$ is a geometric random variable
	with parameter $m_i$ and range $\{1,2,3,\ldots\}$.
Suppose that all $m_i$'s are distinct.
Then we have $\Pr[X=t]=0$ for $t < \ell$, and for all $t\ge \ell$,
\[
\Pr[X=t] = \Big(\prod_{i=1}^\ell m_i\Big)  \cdot  
	\sum_{i=1}^\ell \frac{(1-m_i)^{t-1}}{\prod_{j = 1, j \neq i}^{\ell} (m_j-m_i)} .
\]
}
\begin{lemma} \label{lem:alldistinct}
{\lemalldistinct}
\end{lemma}

We defer the proof of Lemma~\ref{lem:alldistinct} to the appendix.
With the above lemmas, we can derive the following theorem for the computation
	of activation probabilities for the chain graphs.

\begin{theorem} \label{thm:chain}
Let $L = (V_L, E_L)$ be a directed length-$r$ chain and
	let $\{u_0\}$ be the seed set.
Let $u$ be a node in the graph with distance $\ell$ from node $u_0$.
For the activation probability $ap(u,t)$ of node $u$ at time $t$, we have
	$ap(u,t)=0$ if $t< \ell$, and
(a) if all meeting probabilities $m_i$'s of the edges are distinct, then
\begin{equation}\label{eqn:apbasic}
  ap(u, t) = \Big(\prod_{i=1}^\ell m_i\,p_i \Big)  \cdot  \sum_{i=1}^\ell \frac{(1-m_i)^{t-1}}{\prod_{j = 1, j \neq i}^{\ell} (m_j-m_i)} \,;
\end{equation}
(b) if all meeting probabilities $m_i$'s equal to a value $m$, then
\begin{equation}\label{eqn:apequal}
ap(u,t) = m^\ell \cdot (1-m)^{t-\ell}\cdot \left(\prod_{i=1}^\ell p_i \right) \cdot \binom {t-1}{\ell-1}.
\end{equation}
\end{theorem}

Computing activation probabilities by Theorem~\ref{thm:chain}
	has time complexity polynomial to $\log\tau$ for $t\le \tau$.
Moreover, if we want to compute the cumulative probability
	$\sum_{t=0}^{\tau} ap(u,t)$, for Equation~\eqref{eqn:apbasic} it is
	easy to see that the key computation is 
	$\sum_{t=\ell}^{\tau} (1-m_i)^{t-1}=\frac{(1-m_i)^{\ell-1}-(1-m_i)^\tau}{m_i}$,
	which can be done in time polynomial to $\log\tau$; for 
	Equation~\eqref{eqn:apequal}, it is well known that the cumulative
	distribution of a negative binomial random variable can be computed via
	the regularized incomplete beta function, which can also be done in time
	polynomial to $\log\tau$.

Although chain graphs seem to be a rather restricted version of
	arborescences, in actual computations, 
	since the size of the seed set 
	is typically much smaller than the size of the original graph, for any
	node $v$, the number of seeds in this particular $v$'s in-arborescence
	is usually quite small, and hence the chain cases will be common.
In view of the above, this can help reduce the running time of our dynamic
	programming algorithm.
An interesting 
	open problem is to design an algorithm that computes influence spread
	in general in-arborescences with running time polynomial in $\log\tau$.
%% Also note that in the implementation, if the input $t < \ell$, then
%% 	$ap(u,t)$ can be obviously set to zero since the number of time
%% 	steps required for the influence to reach $u$ requires
%% 	at least the path length $\ell$.

%\subsubsection*{When Meeting Probabilities Are Not Necessarily All Distinct}

%%%%%%%%%%%%%%%%%%%%%%%%%%%
%%%%%%%%%%%%%%%%%%%%%%%%%%%
%%%%%%%%%%%%%%%%%%%%%%%%%%%
\eat{
When meeting probabilities are not all distinct.
Consider a 3-edge chain $$u = u_1 \ra u_2 \ra u_3 \ra u_4 = v$$
Suppose 2 out of the 3 meeting probabilities are equal.

\smallskip\noindent\textbf{Case 1:} 
$m(u_1,u_2) = \mathbf{m_1}$,
$m(u_2,u_3) = \mathbf{m_1}$,
$m(u_3,u_4) = \mathbf{m_2}$.
\begin{eqnarray*}
	ap(v,t) &=& \frac{m_1^2 m_2 p_1 p_2 p_3}{(m_1-m_2)^2} \Big[ (1-m_2)^{t-1} - (1-m_2)(1-m_1)^{t-2} - (t-2)(m_1-m_2)(1-m_1)^{t-2} \Big]  \\
	&=& \frac{m_1^2 m_2 p_1 p_2 p_3}{(m_1-m_2)^2} \Big[ (1-m_2)^{t-1} - (1-m_1)^{t-1} - (t-1)(m_1-m_2)(1-m_1)^{t-2} \Big]
\end{eqnarray*}

\smallskip\noindent\textbf{Case 2:} 
$m(u_1,u_2) = \mathbf{m_1}$,
$m(u_2,u_3) = \mathbf{m_2}$,
$m(u_3,u_4) = \mathbf{m_1}$.
$$ap(v,t) = \frac{m_1^2 m_2 p_1 p_2 p_3}{(m_1-m_2)^2} \Big[ (1-m_2)^{t-1} - (1-m_1)^{t-1} - (t-1)(m_1-m_2)(1-m_1)^{t-2} \Big]$$

\smallskip\noindent\textbf{Case 3:} 
$m(u_1,u_2) = \mathbf{m_2}$,
$m(u_2,u_3) = \mathbf{m_1}$,
$m(u_3,u_4) = \mathbf{m_1}$.
$$ap(v,t) = \frac{m_1^2 m_2 p_1 p_2 p_3}{(m_1-m_2)^2} \Big[ (1-m_2)^{t-1} - (1-m_1)^{t-1} - (t-1)(m_1-m_2)(1-m_1)^{t-2} \Big]$$

All of the three cases are the same!
}

\section{MIA Algorithms for IC-M}\label{sec:algo}
The aforementioned greedy approximation algorithm is too inefficient to use in practice as it lacks of a way to efficiently compute influence spread in general graphs (Sec.~\ref{sec:model}). 
To circumvent such inefficiency, we propose two MIA-based heuristic algorithms.
%that leverages arborescences to approximate the local influence regions of nodes.
The first algorithm is {MIA-M} (Maximum Influence Arborescence for IC-M) which uses the dynamic programming in Theorem~\ref{thm:ap} to compute {exact} influence of seeds.
The second one is {MIA-C} (Maximum Influence Arborescence with Converted propagation probabilities) which first estimates propagation probabilities for pairwise users by combining meeting events, influence events, and the deadline $\tau$, and then uses MIA for IC to select seeds.
%Experiments (Sec.~\ref{sec:exp}) show that both MIA-M and MIA-C are effective, efficient, and scalable to large graphs.

Both algorithms first construct a maximum influence in-arborescence (MIIA) for each node in the graph, we calculate influence propagated through these MIIAs to approximate the influence in the original network.

\subsection{The MIA-M Algorithm}\label{sec:miam}
Before describing the algorithm, we first introduce some necessary notations.
For a pair of nodes $u,v$, let $\mathcal{P}(u,v)$ be the set of all paths from $u$ to $v$ in $G$.
Given a path $\P = \langle u = u_1, \dotsc, u_l = v \rangle \in \mathcal{P}(u,v)$, its propagation probability
$$
pp(\P) = \prod_{i=1}^{l-1} p(u_i, u_{i+1}).
$$

Next, we define the \emph{maximum influence path} from $u$ to $v$ to be
$$
\MIP(u,v)=\argmax_{\P \in \mathcal{P}(u,v)} pp(\P).
$$
Note that $\MIP(u,v) = \emptyset$ if $u=v$ or $\mathcal{P}(u,v) = \emptyset$.
In addition, we require at most one $\MIP(u,v)$ for each $u,v$ pair, with ties broken in a consistent way.
To compute MIPs, notice that if we transfer influence probability $p(u,v)$ into edge weight $-\log p(u,v)$, computing $\MIP(u,v)$ is equivalent to finding the shortest path from $u$ to $v$ in $G$, and this can be done efficiently by Dijkstra's algorithm.

For MIA-M, we also introduce the ``augmented'' length $\ell_A(\P)$ of a path $\P$ to take meeting events and the deadline constraint into account.
Consider an edge $(u_i,u_j) \in \P$.
Due to random meeting events, after $u_i$ activates at step $t$, its influence will not propagate to $u_j$ exactly at $t+1$.
Instead, the propagation may take multiple steps and the number of such steps is a random variable $X_{i,j}$, which can also be interpreted as the number of Bernoulli trials needed to get the first meeting between $u_i$ and $u_j$ after $u_i$'s activation (See also Sec~\ref{sec:chain}).
Clearly, $X_{i,j}$ follows the geometric distribution, with success probability $m(u_i,u_j)$, expectation $\frac{1}{m(u_i,u_j)}$, and standard deviation $\frac{\sqrt{1-m(u_i,u_j)}}{m(u_i,u_j)}$.
Here we propose to estimate the value of $X_{i,j}$ by
$$\frac{1}{m(u_i,u_j)} - \frac{\sqrt{1-m(u_i,u_j)}}{m(u_i,u_j)},$$
and define the \emph{augmented path length} $\ell_A(\P)$ of $\P$ to be the sum of all estimated values of the random variables (one per edge) along $\P$:
$$
\ell_A(\P) = \sum_{(u_i,u_j) \in \P} \Biggr( \frac{1}{m(u_i,u_j)} - \frac{\sqrt{1-m(u_i,u_j)}}{m(u_i,u_j)} \Biggr).
$$
We empirically verify that this is a good choice for $\ell_A(\P)$. 
 
\subsubsection*{Constructing Arborescences}
For any node $v$ in $G$, we approximate the influence to $v$ from all $u\in V\setminus\{v\}$ using the \emph{maximum influence in-arborescences (MIIA)} of $v$.
To construct the MIIA rooted at $v$, we first take the union over the maximum influence paths to $v$ over all $u \in V\setminus\{v\}$.
After that, two pruning steps will be done.
First, we remove paths whose propagation probability is below a pre-defined influence threshold $\theta \in (0,1]$, which controls the size of the local influence region and is a trade-off between efficiency and seed set quality.
Second, to take the effect of deadline into account, we eliminate paths whose augmented length is greater than $\tau$.

\begin{definition}[Maximum Influence In-Arborescence]\label{def:miia}
	Given an influence threshold $\theta \in (0,1]$ and a deadline constraint $\tau \in \mathbb{Z}_+$, the maximum influence in-arborescence of any node $v \in V$ is
	\begin{equation*}
	\MIIA(v,\theta) = \cup_{u\in V, pp(\MIP(u,v))\geq \theta,\ell_A(\MIP(u,v))\leq \tau} \MIP(u,v)
	\end{equation*}
\end{definition}

The full MIA-M is described in Algorithm~\ref{alg:mia-m}, where $MG(u) = \sigma_\tau(S\cup\{u\}) -\sigma_\tau(S)$ is the marginal influence of $u$ w.r.t to seed set $S$, $MG(u,v)$ is the marginal influence of $u$ on a specific $v$, and $realized(v)$ is the cumulative influence realized on $v$ by $S$.
Also, for each $u \in V$, $\InfSet(u) = \{v \in V : u \in \MIIA(v,\theta)\}$.

After constructing $\MIIA(v,\theta)$ and using Theorem~\ref{thm:ap} to obtain $\sigma_\tau(\{v\})$ for all $v \in V$ (lines 4-10), the algorithm selects $k$ seeds iteratively in a greedy manner, and uses Theorem~\ref{thm:ap} to update the marginal gain of nodes in related MIIAs (lines 11-20).
Specifically, after $u$ is picked as a seed, the activation probability of all $v \in \InfSet(u)$ goes up, and thus we need to update the marginal gain of all $w \in \MIIA(v)$, $\forall v \in \InfSet(u)$.

%The initialization part (lines 3-7) computes $\MIIA(v,\theta)$ and $MG(v)=\sigma_\tau(\{v\})$ for all $v$.
%The core part of MIA-M (lines 8-17) repeatedly selects $k$ seeds and update marginal influence of related nodes after picking each seed.  

\begin{algorithm}[h!]
\small
\caption{MIA-M ($G=(V,E)$, $k$, $\theta$, $\tau$)}\label{alg:mia-m}
$S \gets \emptyset$\;
$\forall v \in V, MG(v) \gets 0$ and $realized(v) \gets 0$\;
$\forall v \in V, \MIIA(v,\theta) \gets \emptyset$ and $\InfSet(v) \gets \emptyset$\;

\ForEach{$v \in V$}{
	Compute $\MIIA(v,\theta)$ (Definition~\ref{def:miia})\;
	\ForEach{$u \in \MIIA(v,\theta)$}
	{
		$\InfSet(u) \gets \InfSet(u)\cup\{v\}$\;
		Compute $ap(v,t,\{u\},\MIIA(v,\theta))$, $\forall t \le\tau$ (Theorem~\ref{thm:ap})\;
		$MG(u,v) \gets \sum_{t=0}^\tau ap(v,t,\{u\},\MIIA(v))$\;
		$MG(u)\gets MG(u) + MG(u,v)$\;
	}
}
\For{$i = 1 \rightarrow k$}
{
	$u \gets \argmax_{v \in V \setminus S} MG(v)$\; 
	$S \gets S \cup \{u\}$\;
	\ForEach{$v \in \InfSet(u)$}
	{
		$realized(v)$ += $MG(u,v)$\;
		\ForEach{$w \in \MIIA(v,\theta)$}
		{
			Compute $ap(v,t,S\cup\{w\},\MIIA(v,\theta))$, $\forall t\le\tau$ (Theorem~\ref{thm:ap})\;
			$MG_{new}(w,v) \gets \big[\sum_{t=0}^\tau ap(v,t,S\cup\{w\},\MIIA(v,\theta))\big] - realized(v)$\;
			$MG(w)\gets MG(w) + MG_{new}(w,v) - MG(w,v)$\;
			$MG(w,v) \gets MG_{new}(w,v)$\;
		}
	}
}
\end{algorithm}

\smallskip\noindent\textbf{Time Complexity:}
Let $n_{m\theta} = \max_{v\in V} |\MIIA(v,\theta)|$, and $n_{s\theta} = \max_{v \in V} |\InfSet(v)|$.
Also suppose that the maximum running time to compute $\MIIA(v,\theta)$ for any $v \in V$ by Dijkstra's algorithm is $t_{m\theta}$.
Thus, MIA-M runs in $O(|V|(t_{m\theta}+n_{m\theta}\tau^3) + k n_{m\theta} n_{s\theta} (n_{m\theta} \tau+\log|V|))$.

\subsection{The MIA-C Algorithm}\label{sec:miac}
We now discuss our second algorithm, MIA with Converted propagation probability (MIA-C).
It consists of two steps.
First, for each $(u,v) \in E$, we estimate a converted propagation probability
	$p_c(u,v)$ that incorporates meeting probability $m(u,v)$, 
	influence probability $p(u,v)$, and deadline $\tau$, with the intention
	to \textsl{simulate} the influence spread under the IC-M model
       in the original IC model.
%% we estimate the \emph{actual propagation probability} from $u$ to $v$, denoted as $p_\tau(u,v)$, which is the probability that $u$ would meet and activate $v$ within $\tau$ steps under the IC-M model.
%% Note that if $p_\tau(u,v)$ is in its exact form for every edge, then our problem is nothing but the original influence maximization problem under IC.
Second, after obtaining all $p_c(u,v)$, we treat these converted probabilities as parameters for the IC model and run the MIA algorithm proposed for IC to select $k$ seeds. 

In the IC-M model with deadline $\tau$, the value of $p_c(u,v)$ depends on $p(u,v)$, $m(u,v)$, and $\tau$.
%% That is, $p_\tau(u,v) = f(p(u,v),m(u,v),\tau)$, where $f$ is the transfer function.
%% Since the number of steps $u$ needs to propagate influence to $v$ is a random variable $X_{u,v}$ (as mentioned in Sec~\ref{sec:miam}) whose exact value is not known before the cascade, it is difficult to calculate $p_\tau(u,v)$ exactly.
%% Therefore, as a fast and intuitive heuristic,
%% %it is reasonable to first estimate the number of attempts $u$ tries to meet $v$, and use this to approximately compute $p'(u,v)$.
We use the following conversion function to obtain $p_c(u,v)$:
\begin{equation}\label{eq:conver}
	p_c(u,v) = p(u,v) \cdot [1-(1-m(u,v))^{\beta}],
\end{equation}
where $\beta \in [1,\tau]$ is the parameter used to estimate the number of
	meeting attempts.
If $\beta$ is $1$, $p_c(u,v) = p(u,v)\cdot m(u,v)$, in which case we are pessimistic that $u$ has only one chance to meet $v$ (the minimum possible, assuming $u$ itself activates before $\tau$).
On the other hand, if $\beta$ is $\tau$, $p_c(u,v) = p(u,v)\cdot [1-(1-m(u,v))^{\tau}]$, for which we are optimistic that $u$ has $\tau$ chances to meet $v$ (the maximum possible).
To achieve a balanced heuristic, we let $\beta = \frac{\tau}{2}$ for all pairs of $u,v$, and experiments show that this estimation turns out to be more effective than other choices in most cases.

After the probability conversion step, we utilize MIA (Algorithm 4, \cite{ChenWW10}) to find the seed set, making MIA-C take the advantage of updating marginal gains of nodes in an extremely efficient manner.
% due to linear relationships between activation probabilities.
%% This is because in MIA, it is shown that there is a linear relationship between $ap(u)$ and $ap(v)$ where $u \in \MIIA(v)$: $ap(v) = \alpha_v(u) \cdot ap(u) + c_v(u)$, where $\alpha_v(u)$ and $c_v(u)$ are constants independent of $ap(v)$.

The time complexity of converting probabilities is $O(|E|)$, and second part of MIA-C has the same time complexity as MIA, which is $O(|V| t_{m\theta} + k n_{m\theta} n_{s\theta}(n_{m\theta} + \log |V|))$.

\section{LDAG Algorithms for LT-M}\label{sec:ldag}
In this section, we present an dynamic programming
	algorithm that computes exact influence spread
	in directed acyclic graphs (DAGs), and use
	it to develop our heuristic solution based on
	DAGs.
Chen et al.~\cite{ChenYZ10} proved that computing
	exact influence spread in general graphs for any
	node set is \SPhard.
Now we show that in directed acyclic graphs, computing
	influence can be done in time linear to the size of
	the graph.

\subsection{Fast Influence Computation in DAGs}

Consider an acyclic subgraph
	$D=(V,E)$.
Let $S\subseteq V$ be the seed set.
For any $v\in V$, let $ap_D(u,t\mid S)$ be the {\em activation
	probability} of $u$ and step $t$, that is, the probability
	that $u$ is activated right at step $t$ in DAG $D$
	under the LT-M model, given the seed set $S$.
When the notations are clear from the context, we write
	$ap(v,t)$ for short.
By definition, $ap(v,t) = 1$ if $v\in S \wedge t=0$.
Similarly, $ap(v,t) = 0$ if $v\in S\wedge t>0$, or
	$v\not\in S\wedge t=0$.
These cases form the basis for the recursion
	that we will develop for the dynamic programming
	algorithm, which computes $ap(v,t)$ when
	$v\not\in S\wedge t>0$.

The following theorem shows the important linear
	property of activation probability in a DAG.
The proof for this theorem also requires Lemma~\ref{lemma:le}.
\begin{theorem}\label{thm:ltap}
For any $v\in V\setminus S$, and any $t\in[1,\tau]$,
	the {\em activation probability} of $v$ at $t$
	is
\begin{equation}\label{eq:ltap}
ap(v,t) = \sum_{u\in\Nin(v)}b(u,v)\sum_{t'=0}^{t-1} ap(u,t')\cdot m(u,v)\cdot (1-m(u,v))^{t-t'-1}.
\end{equation}
\end{theorem}

\begin{proof}
By Lemma~\ref{lemma:le}, the event that a certain node $v$
	becomes activate at a certain step $t$ has the same probability
	to happen under the LT-M model and the LE-M model.
Let $\ev_{u,v}$ be the event that $v$ selects edge $(u,v)$
	as a live edge.
By the protocol of the live-edge selection process, we have 
	$\Pr[\ev_{u,v}] = b(u,v)$.
Let $R(v,t)$ be the event that $v$ is reached by $S$ at step $t$. 
Then, for all $v\in V\setminus S$,
\begin{equation}\label{eq:ap1}
ap(v,t) = \sum_{u\in V\setminus\{v\}} \Pr[\ev_{u,v}]\cdot \Pr[R(v,t)\mid \ev_{u,v}].
\end{equation}

In order to activate $v$ at $t$, the selected $u$ must be
	already active before or at $t-1$.
Furthermore, the events that $u$ activates at possible
	steps $0,1,\dotso$ are mutually exclusive.
Hence, the probability that event $R(v,t)$ happens conditioned on
	$\ev_{u,v}$ is
\begin{equation}\label{eq:ap2}
\Pr[R(v,t)\mid \ev_{u,v}] = \sum_{t'=0}^{t-1} ap(u,t')\cdot m(u,v)\cdot (1-m(u,v))^{t-t'-1}.
\end{equation}

Putting Eq.~\eqref{eq:ap1} and \eqref{eq:ap2} together, we obtain Eq.~\eqref{eq:ltap}.
This completes the proof.
\end{proof}

A directed application of Theorem~\ref{thm:ltap} gives us
	a dynamic programming method to compute
	activation probabilities.
\eat{
Notice, however, that time complexity of computing
	$ap(v,t)$, given as input a particular pair of
	$v\in V\setminus S$ and $t\le \tau]$,
	is exponential to the size of the input: $\Theta(\log\tau)$ bits.
In principle, this does not affect efficiency much
	since in general $\tau$ is much smaller than
	the size of the graph.
}

\subsection{The LDAG-M Heuristic Algorithm}
In this subsection, we propose a DAG-based heuristic
	algorithm, called LDAG-M, which leverages the proposed dynamic
	programming approach to compute influence efficiently.
The algorithm first constructs a local directed acyclic
	graph (LDAG) for each node $v$ in the graph.
These LDAGs are small subgraphs of the original network.
Then, LDAG-M uses influence propagated through
	these LDAGs to approximate the influence propagated
	in the original network.

\subsubsection*{Local DAG Construction}

Let $Inf^G_\tau(v,\{u\})$ be the probability that $v$ gets
	activated in $G$ when $u$ is the only seed.
Similarly, let $Inf^D_\tau(v,\{u\})$ be such probability in a
	DAG $D$.
Given an influence threshold $\lambda$, a node $v$,
	we wish to find a DAG $D=(U,F)$ such that
	$v\in U\subseteq V$, $F\subseteq E$, and $Inf^D_\tau(v,\{u\}) \ge \lambda, \forall u\in U$.
There are possibly more than one DAGs satisfying the above
	conditions, but ideally, we want to find an optimal
	DAG $D^*$ such that $\sum_{u\in U^*} Inf^{D^*}_\tau(v,\{u\})$
	is the largest among all possible choices.
Unfortunately, this is \NPhard{}~\cite{ChenYZ10}.
Therefore, we construct LDAGs heuritically in a greedy fashion.

We first run the Find-LDAG algorithm due to~\cite{ChenYZ10} (Algorithm~\ref{alg:ldag}).
It starts from an empty node set $U$ and an empty edge set $F$.
All $Inf(v,\{u\})$ are initialized to zero, except that
	$Inf(v,\{u\}) = 1$.
It repeatedly picks the node $w$ that has the
	largest influence value to $v$, and add
	edges from $w$ to existing nodes in $U$ into $F$.
Then it adds $w$ into $U$, and updates the influence values
	of all of $w$'s in-neighbors.
Find-LDAG terminates when no new node has influence
	value of at least $\lambda$.

\begin{algorithm}[t!]
\caption{Find-LDAG ($G=(V,E)$, $v$, $\lambda$)}\label{alg:ldag}
$U \gets \emptyset$; $F \gets \emptyset$\;
\ForEach {$u\in V\setminus\{v\}$} {
 	$Inf(v,\{u\})\gets 0$\;
}
$Inf(v,\{v\})\gets 1$\;
\While {$\max_{u\in V\setminus U} Inf(v,\{u\}) \ge \lambda$} {
	$w \gets \argmax_{u\in V\setminus U}Inf(v,\{u\})$\;
	$F \gets F\cup \{(w,x): x\in U\}$\;
	$U \gets U\cup \{w\}$\;
}
\ForEach{$u \in \Nin(w)$} {
	$Inf(v,\{u\})$ += $b(u,w)\cdot Inf(v,\{w\})$\;
}
$LDAG(v,\lambda) \gets (U,F)$ and output it\;
\end{algorithm}

This algorithm suffices for the original influence maximization
	problem under the LT model.
However, in our case, it is very important for us to take meeting
	events and the deadline constraint $\tau$ into account.
To this end, we need to post-process the LDAGs find
	by Algorithm~\ref{alg:ldag}.
First, for each edge $(u,w)\in F$, we obtain transfer the meeting probability $m(u,w)$
	into a distance weight
$\frac{1}{m(u,w)} - \frac{\sqrt{1-m(u,w)}}{m(u,w)}$.
Then, for each $u\in U\setminus\{v\}$, we compute its shortest path
	to $v$ using Dijkstra's algorithm, and let this path be $\P(u,v)$.
Next, we calculate augmented path length of $\P(u,v)$,
	which is $\ell_A(\P) = \sum_{(u_i,u_j) \in \P(u,v)} \Big( \frac{1}{m(u_i,u_j)} - \frac{\sqrt{1-m(u_i,u_j)}}{m(u_i,u_j)} \Big).$
The pruning condition is that if $\ell_A(\P) > \tau$, then node $u$
	and all corresponding edges are removed from $LDAG(v,\lambda)$.
Notice that for a DAG, we can do a topological sort on all nodes,
	so the order of pruning can be determined by that,
	with the $v$ sorted first.
This process is reflected in line 4 of Algorithm~\ref{alg:ldagm}.

\subsubsection*{Seed Set Selection and Full LDAG-M Algorithm}

The full LDAG-M algorithm first uses the approach described
	above to find out a suitable local DAG for all nodes
	in the graph (lines 3-4), then it greedily select at most $k$ seeds,
	given the input $G=(V,E)$ and the budget $k$.
The pseudo-code of LDAG-M is given in Algorithm~\ref{alg:ldagm}.  
For each node $u$, we also maintain a data structure called $InfSet(u) := \{v\in V\colon u\in LDAG(v,\lambda) \}$,
	a value $MG(u)$ which denotes the incremental influence
	by adding $u$ to the current seet set $S$, and a value
	$realized(v)$ denotes the cumulative influence realized on
	$v$ by the seed set $S$.
After constructing LDAGs and InfSets, we use Theorem~\ref{thm:ltap}
	to compute obtain $\sigma_\tau(\{u\}$ for all $u$ in the graph
	(lines 7-9).
Line 10-19 iteratively pick $k$ seeds in a greedy manner.
After selecting a new seed $u$, we need to update the incremental 
	influence of all $w$ in $LDAG(v,\lambda), \forall v\in InfSet(u)$.

\eat{
\smallskip\noindent\textbf{The LDAG-C Algorithm.}
Similar to the MIA-C heuristic described in Section~\ref{sec:miac},
	we can also convert the influence weights in the LT-M model
	using the conversion function \eqref{eq:conver}
\begin{equation}\label{eq:converlt}
	b_c(u,v) = b(u,v) \cdot [1-(1-m(u,v))^{\beta}],
\end{equation}
}

\begin{algorithm}[t!]
\caption{LDAG-M ($G=(V,E)$, $k$, $\lambda$)}\label{alg:ldagm}
$S \gets \emptyset$;
$\forall v\in V: MG(v)\gets 0$;
$realized(v)\gets 0$\;
\ForEach {$v\in V$} {
 	Compute $LDAG(v, \lambda) \quad$ (Algorithm~\ref{alg:ldag})\;
	Post-process (prune) $LDAG(v, \lambda)$\;
	\ForEach {$u\in$ $LDAG(v,\lambda)$} {
        	$InfSet(u) \gets v$\;
		Compute $ap(v,t|\{u\}), \forall t\le \tau\,$ (Theorem~\ref{thm:ltap})\;
		$MG(u,v) \gets \sum_{t=0}^\tau ap(v,t|\{u\})$\;
		$MG(u)$ += $MG(u,v)$\;
	}
}
\For {$i= 1\to k$} {
 	$u\gets \argmax_{v\in V\setminus S} MG(v)$\;
	$S\gets S\cup \{x\}$\;
	\ForEach{$v\in InfSet(u)$} {
		$realized(v)$ += $MG(u,v)$\;
		\ForEach{$w\in LDAG(v)$} {
			Compute $ap(v,t | S\cup\{w\}), \forall t \le \tau$ (Theorem~\ref{thm:ltap})\;
			$MG_{new}(w,v) \gets  \sum_{t=0}^\tau ap(v,t | S\cup\{w\}) - realized(v)  $\;
			$MG(w)$ += $MG_{new}(w,v) - MG(w,v)$\;
			$MG(w,v)\gets MG_{new}(w,v)$\;
		}
	}
}
\end{algorithm}

\section{Empirical Evaluations}\label{sec:exp}
\begin{table} %[p!h!t!]
	\centering
%	\scriptsize
	\begin{tabular}{|l | c | c | c | c |}
		\hline
		\textbf{Dataset} & \hept & \wiki & \epi & \dblp \\ \hline
		Number of nodes & 15K & 7.1K & 75K & 655K \\ \hline
		Number of edges & 62K & 101K & 509K & 2.0M  \\ \hline
	    Average degree & 4.12 & 26.6 & 13.4 & 6.1  \\ \hline
	    Maximum degree & 64 & 1065 & 3079 & 588  \\ \hline
	    \#Connected components & 1781 & 24 & 11  & 73K  \\ \hline
		Largest component size & 6794 & 7066 & 76K & 517K  \\ \hline
	 	Average component size & 8.55 & 296.5 & 6.9K  & 9.0  \\ \hline
	\end{tabular}
\caption{Statistics of Real-world Networks.}
	 \label{table:dataset}
	\vspace{-3mm}
\end{table}

We conduct experiments on four real-world datasets to evaluate MIA-M and MIA-C, and compare them to a few other algorithms in terms of seed set quality and running time.
All experiments are conducted on a server running Microsoft Windows Server 2008 R2 with 2.33GHz Quad-Core Intel Xeon E5410 CPU and 32G memory.

\subsection{Experiment Setup}

\paragraph{Dataset Preparation.}
The statistics of the datasets are summarized in Table~\ref{table:dataset}.
NetHEPT is a standard dataset in this area: It is a collaboration network extracted from the High Energy Physics Theory section (1991 to 2003) of the arXiv e-print repository (http://www.arxiv.org/).
The network data is publicly available at http://research.microsoft.com/en-us/people/weic/projects.aspx.
DBLP (http://www.informatik.uni-trier.de/\~ley/db/) is a much larger collaboration network from the DBLP computer science bibliography server maintained by Michael Ley.
Nodes in both datasets represent authors, and if $u$ and $v$ collaborated at least once, we draw direct arcs $(u,v)$ and $(v,u)$.
Note that edges in the NetHEPT and DBLP graphs may carry multiplicity greater than $1$, because two authors might co-author more than one papers.

WikiVote is a who-votes-on-whom network extracted from Wikipedia, the free online encyclopedia where users can interact with each other when they co-edit the same entries.
If $v$ voted on $u$ for promoting $u$ to adminship, we draw a directed arc $(u,v)$ to reflect $u$'s influence on $v$.
Epinions is a who-trusts-whom social network from the Epinions consumer review site (http://www.epinions.com/).
We draw a directed arc $(u,v)$ if $v$ expressed her trust in $u$'s reviews explicitly on the Website.
Both the WikiVote and Epinions network data are can be obtained from the Stanford Network Analysis Project Website (http://snap.stanford.edu/data/).

\paragraph{Graph Parameters (Models for Assigning Influence and Meeting Probabilities).}
Influence probabilities are assigned using the Weighted Cascade (WC) model
	proposed in Kempe et al.~\cite{kempe03}:
For WikiVote and Epinions, $p(u,v) = 1/d^{in}(v)$ where
	$d^{in}(v)$ is the in-degree of $v$.
For NetHEPT and DBLP, $p(u,v) = A(u,v) / A(v)$ where $A(u,v)$ is the number
	of papers in which $u$ and $v$ were co-authors, and $A(v)$ is the
	number of papers that $v$ published in total. 
We also do experiments on the Trivalency (TV) model proposed in
	Chen et al.~\cite{ChenWW10}: On every edge $(u,v)$, we choose its edge
	probability uniformly at random from the set $\{0.001, 0.01, 0.1\}$.

For meeting probabilities, it is reasonable to deem that the more friends
	an individual $u$ has, the smaller the chance that $u$
	could meet a certain friend is.
Therefore, we assign each edge $(u,v)$ its meeting probability
	$m(u,v) = \frac{c}{d^{out}(u)+c}$ , where $d^{out}$ is the out-degree of $u$
	and $c$ is a constant chosen to be $5$ here.
In addition, we test on cases where for every edge $(u,v)$, $m(u,v)$ is
	chosen uniformly at random from the set
	$\{0.2, 0.3, 0.4, 0.5, 0.6, 0.7, 0.8\}$.
%following literature in contact-aware wireless sensor networks~\cite{}.

\paragraph{Algorithms Compared.}
We evaluate MIA-M, MIA-C, the greedy algorithm (Greedy) and the other two: Degree and MIA.
For Greedy, we apply CELF and run Monte Carlo for 10000 times as in Kempe et al.~\cite{kempe03}.
%we run Monte Carlo $10000$ times as in Kempe et al.~\cite{kempe03} and the CELF optimization.
Degree is a heuristic based on the notion of degree centrality that considers high degree nodes influential.
It outputs the top-$k$ highest out-degree nodes as seeds (Kempe et al.~\cite{kempe03}).
We then test MIA, one of the state-of-the-art heuristic algorithms for the standard IC model.
For the purpose of comparisons, we let MIA select seeds disregarding meeting probabilities and the deadline constraint entirely, i.e., treating $m(u,v)=1$ for all edges and $\tau = |V|$.
We also test on the Prefix-excluding MIA (PMIA) algorithm, which is a variant of MIA~\cite{ChenWW10}.
Since the results are similar, we omit it here.
MIA-M, MIA-C, and MIA all use $1/320$ as the influence threshold $\theta$, as recommended by Chen et al.~\cite{ChenWW10}.
For MIA-C, we choose $\beta = \frac{\tau}{2}$ since it gives more stable performance (compared to $1$ and $\tau$) in most cases.

\begin{figure*}
\begin{tabular}{ccc}
    	\includegraphics[width=.45\textwidth]{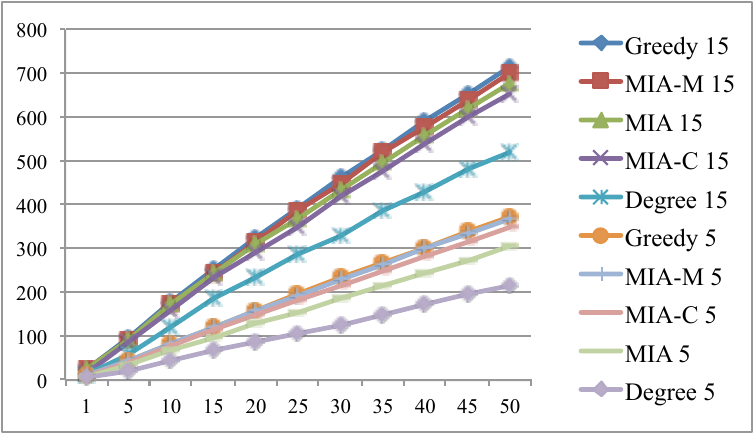}&
	\includegraphics[width=.45\textwidth]{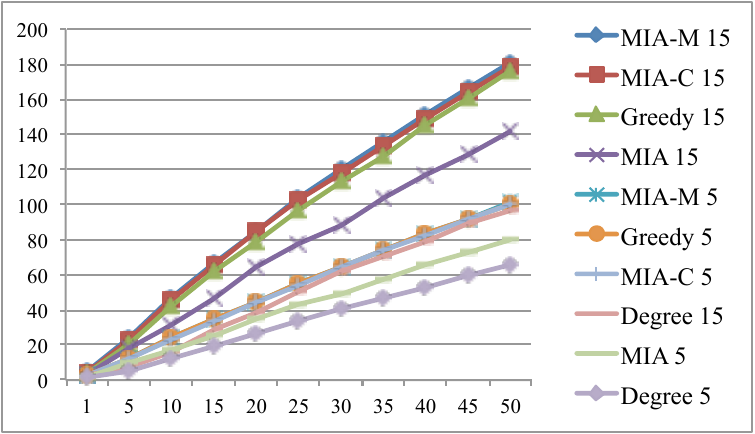}&\\
	(a) NetHEPT  & (b) WikiVote &\\
    	\includegraphics[width=.45\textwidth]{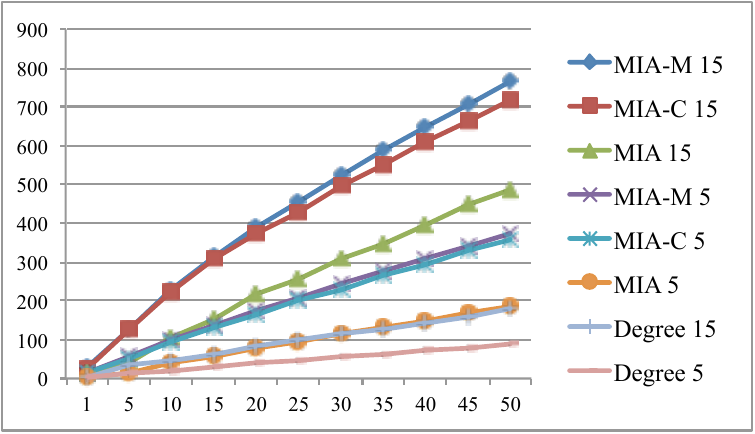}&
    	\includegraphics[width=.45\textwidth]{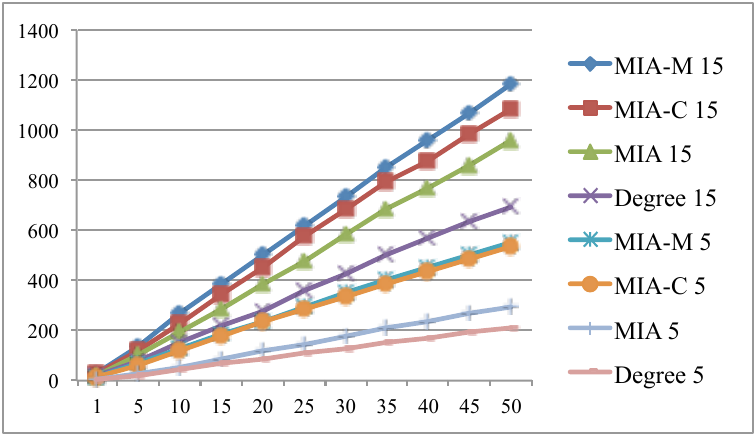}\\
	 (c) Epinions & (d) DBLP  \\
\end{tabular}
\caption{\small Influence spread (\#nodes, Y-axis) against seed set size (X-axis) on graphs with weighted meeting probabilities. }
\label{fig:infWeighted}
\end{figure*}

\begin{figure*}
\begin{tabular}{ccc}
    \includegraphics[width=.45\textwidth]{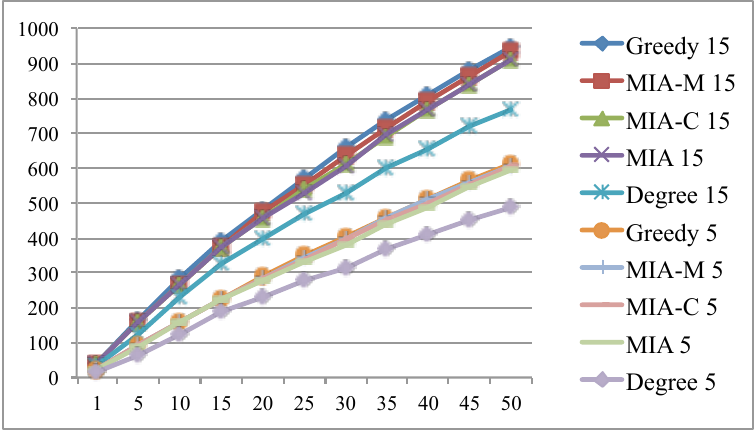}&
    \includegraphics[width=.45\textwidth]{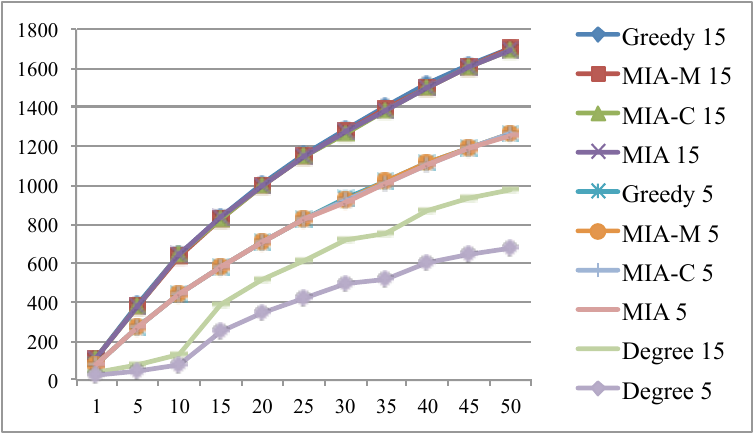}&\\
	(a) NetHEPT  & (b) WikiVote &\\ 
    \includegraphics[width=.45\textwidth]{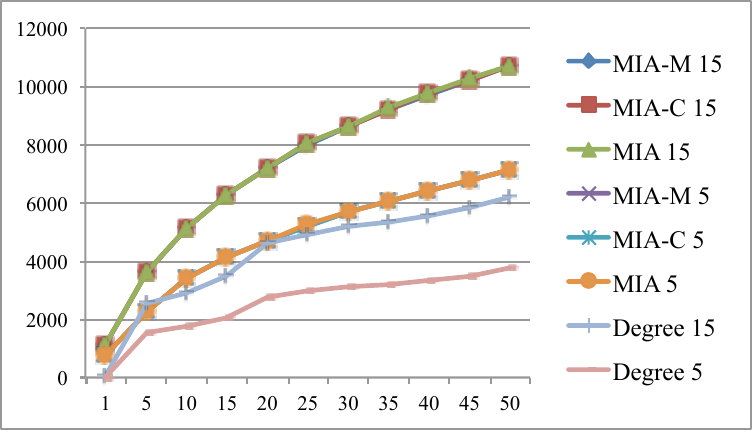}&
    \includegraphics[width=.45\textwidth]{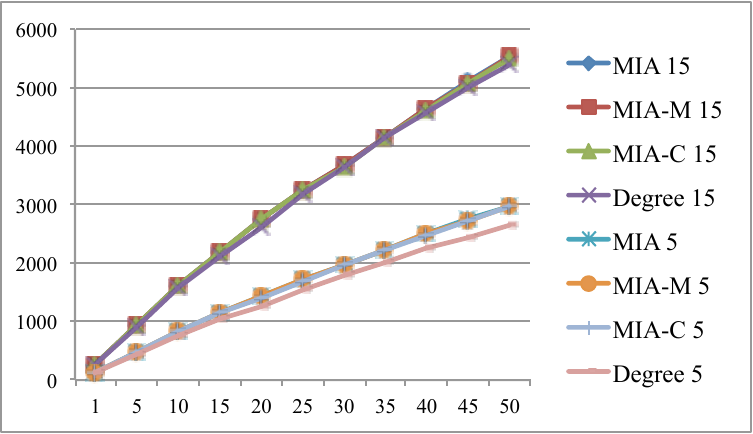}\\
	(c) Epinions & (d) DBLP  \\
\end{tabular}
\caption{\small Influence spread (\#nodes, Y-axis) against seed set size (X-axis) on graphs with uniform random meeting probabilities. }
\label{fig:infRandom}
\end{figure*}

\subsection{Results and Analysis}
We compare the five algorithms on quality of seeds sets and running time.
The deadline $\tau$ is set to $5$ (relatively short time horizon) and $15$  (relatively long time horizon) in all results reported.
Greedy is too slow to finish on Epinions and DBLP within a reasonable amount of time (3 days).
%Since Greedy is lack of efficiency and scalability, it only managed to finish on NetHEPT and WikiVote within a reasonable amount of time.

\paragraph{Quality of Seed Sets.}
The quality of seed sets is evaluated based on the expected influence spread achieved.
To ensure fair and accurate comparisons, we run MC simulations 10000 times to get the ``ground truth'' influence spread of all seed sets obtained by various algorithms.
Fig.~\ref{fig:infWeighted} and \ref{fig:infRandom} illustrate influence spread achieved on datasets with weighted and uniform random meeting probabilities, respectively.

On graphs with weighted meeting probabilities, except for Greedy, MIA-M has the highest seed set quality, while MIA-C is the second best in most test cases.
MIA-M performs consistently better than Degree and MIA, e.g., on Epinions, the influence of 50 seeds by MIA-M is $99.4\%$ ($\tau=5$) and $53.6\%$ ($\tau=15$) higher than those by MIA.
%For example, on Epinions with $\tau=15$, MIA-M, MIA-T, Degree, and PMIA achieve $764$, $717$, $180$, $497$, respectively.
On NetHEPT and WikiVote, MIA-M produces seed sets with equally good quality as Greedy does, e.g., on WikiVote, when $\tau=5$ they both achieve influence spread of 101; when $\tau=15$, MIA-M (181) even achieves $3\%$ higher than Greedy (175).

When meeting probabilities are assigned uniformly at random, seed sets by MIA-M, MIA-C, and MIA tend to have matching influence, all being close to Greedy and better than Degree.
Most often, MIA-M is marginally better than MIA-C and MIA.
The reason why MIA catches up is that it picks seeds assuming all $m(u,v)=1$ and $\tau=|V|$, and in expectation those seeds will still have high influence under uniform random $m(u,v)$'s, because the expectation is taken over the probability space of all meeting events.
Also, when $\tau$ is large, the time-critical effect of the deadline is diminishing, so MIA tends to performs better with large $\tau$.

In reality, however, meeting probabilities between individuals in social networks may be quite different from being uniform random, and over all test cases it can be seen that MIA-M and MIA-C are more stable than MIA.
For certain meeting probabilities, MIA has poor performance.

\begin{table}
%\hspace{-3mm}
	\centering
%	\scriptsize
	\begin{tabular}{|l|c|c|c|c|c|c|c|c|}
 	 \hline
 	\multirow{2}{*}{Algorithm} &
  	\multicolumn{2}{|c|}{NetHEPT} & 
  	\multicolumn{2}{|c|}{WikiVote} &
 	\multicolumn{2}{|c|}{Epinions} & 
	\multicolumn{2}{|c|}{DBLP} \\ \cline{2-9} 
	& $5$ & $15$ & $5$ & $15$ & $5$ & $15$ & $5$ & $15$  \\ \hline 
	Greedy & 40m & 1.3h & 22m & 28m & - & - & - & - \\ \hline 
    MIA-M & 1.6s & 15s & 7.9s & 43s & 47s & 5.1m & 6.6m & 10m  \\ \hline
    MIA-C & 0.3s & 0.3s & 0.4s & 0.5s & 2.7s & 3.3s & 24s & 33s  \\ \hline
    MIA & 0.3s & 0.3s & 1.4s & 1.4s & 12s & 13s & 40s & 41s \\ \hline
    %PMIA & 0.43 & 0.44 & 7.1 & 7.1 & 44.2 & 45.7 & 41.5 & 42.7  \\ \hline
\end{tabular}
	\caption{Running Time (Weighted Meeting Probability)}
	\label{table:timeWeighted}
\end{table}

\begin{table}\label{table:timeRandom}
%\hspace{-3mm}
	\centering
%	\scriptsize
	\begin{tabular}{|l|c|c|c|c|c|c|c|c|}
 	 \hline
 	\multirow{2}{*}{Algorithm} &
  	\multicolumn{2}{|c|}{NetHEPT} & 
  	\multicolumn{2}{|c|}{WikiVote} &
 	\multicolumn{2}{|c|}{Epinions} & 
	\multicolumn{2}{|c|}{DBLP} \\ \cline{2-9} 
	& $5$ & $15$ & $5$ & $15$ & $5$ & $15$ & $5$ & $15$  \\ \hline 
	Greedy & 44m & 1.5h & 1.1h & 3.2h & - & - & - & - \\ \hline 
    	MIA-M & 3.8s & 21.4s & 28.7s & 2.5m & 3.3m & 12.4m & 7.3m & 14.3m  \\ \hline
    	MIA-C & 0.2s & 0.2s & 0.57s & 1.23s & 4.9s & 10.2 s& 32.5s & 45.6s  \\ \hline
    	MIA & 0.4s & 0.5s & 6.5s & 6.5s & 37.9s & 39.7s & 38.9s & 43.1s  \\ \hline
\end{tabular}
\caption{Running Time fo Random Meeting Probabilities}
\label{table:timeRand}
\end{table} 
\paragraph{Running Time.}
We demonstrate the running time results on weighted meeting probability datasets in Table~\ref{table:timeWeighted} and the results on the uniform random cases in Table~\ref{table:timeRand}
Greedy takes 0.5 to 1.3 hours to finish on NetHEPT and WikiVote, and fails to complete in a reasonable amount of time (three days) on Epinions and DBLP with $\tau = 5$.
Degree finishes almost instantly in all test cases so it is not included in the table.

MIA-C and MIA are three orders of magnitude faster than Greedy, since both benefit from the linearity rule of activation probabilities when updating marginal gains~\cite{ChenWW10}.
MIA-C is more efficient because its converted probabilities are smaller than the original influence probabilities used in MIA, and hence arborescences are smaller for MIA-C under the same influence threshold ($1/320$).
MIA-M is two orders of magnitude faster than Greedy, and is scalable to large graphs like Epinions and DBLP.
It is slower than MIA-C and MIA because its dynamic programming procedure computes activation probabilities associated with steps, and hence is not compatible with the linearity rule of activation probabilities.

\begin{figure}[htb]
\centering
\includegraphics[scale=0.7]{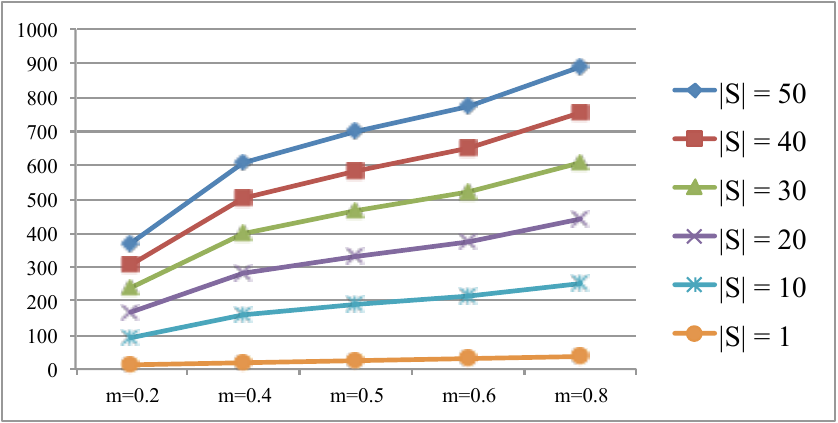}
\caption{\small Influence spread vs. meeting probabilities on NetHEPT.}
\label{fig:meetProb}
\end{figure}

\paragraph{Effects of Deadline Constraints and Meeting Probabilities.}
It can be seen from both Fig.~\ref{fig:infWeighted} and \ref{fig:infRandom} that as $\tau$ increases from 5 to 15, seed sets by all algorithms obtain higher influence spread, which is intuitive to see.
We also test other values of $\tau$, such as 10 and 20, and since the trend is the same, to avoid densely clustered figures we do not include them here.
For meeting probabilities, we conduct five test cases on NetHEPT, running Greedy with meeting probability $0.2, 0.4, 0.5, 0.6, 0.8$ for each edge in the graph. 
The results in Fig.~\ref{fig:meetProb} show that as meeting probabilities increase, the influence spread of the seed set also go up.

\section{Conclusion and Discussions}\label{sec:discuss}
In this paper, we extend the IC and LT models, and their generalization, the Triggering Set model, to include time-delayed influence
	diffusion and we consider the time-critical influence maximization problem.
We prove the submodularity of the influence function under these time-delayed models, and propose fast heuristics to
	solve the problem.
There are a number of extensions and future directions on time-critical
	influence maximization. 

\eat{
One problem is to look into more efficient computation of influence spread
	in tree structures 
	with deadline constraint $\tau$ and meeting probabilities that is 
	polynomial to $\log \tau$.
We have obtained partial results on chain graphs, a special
	class of trees, which could already improve the running
	time of MIA-M.
}

One extension is to use login probabilities to model time-delayed 
	influence diffusion, which could fit better 
	into online social networks.
Specifically, each user has a probability of entering the system, and only after this action,
	the user could be influenced by her friends who are already activated.
Incorporating these login probabilities into the current models turns out to be more
	challenging than incorporating meeting probabilities, because it introduces
	dependency in activation events.
We have obtained partial results using
	more complicated dynamic programming methods to deal with
	this case, which we include in Appendix~\ref{sec:login}.

\section*{Acknowledgement}
We thank Yang Yuan (Peking University) for giving a proof for Lemma~\ref{lemma:z+1} in the appendix.
We also thank Dr.\ Nicolas J.A.\ Harvey (Department of Computer Science, UBC) and Dr.\ Ruben H.\ Zamar (Department of Statistics, UBC) for helpful conversations about results in Sec~\ref{sec:chain}.

\bibliographystyle{abbrv}
\bibliography{singlebib}

%% \clearpage

\appendix\section*{Appendix}
\section{Analysis on the summation of geometric random variables}\label{sec:proof}

In this appendix, we provide analysis of the distribution of summation of
	geometric random variables.
We first provide a proof for the negative binomial distribution for
	completeness.

\paragraph{Lemma~\ref{lem:negbinomial} (Negative Binomial Distribution, 
	re-stated).}
{\lemnegbinomial}
\begin{proof}
We prove the lemma by an induction on $\ell$.
For the base case when $\ell=1$, $X$ is just a geometric random
	variable with parameter $m$, and thus $\Pr[X=t] = (1-m)^{t-1}\,m$.
For the induction step, assume that the lemma
	holds when $\ell = z$ for some $z$.
Let $X' = \sum_{i=1}^{z} X_i$ and $X = \sum_{i=1}^{z+1} X_i$.
It is clear that when $t < z+1$, $\Pr[X=t] = 0$.
Then for the case of $\ell = z+1$ and $t \ge z+1$,
	we have
\begin{align*}
\Pr[X = t]
	&= \sum_{t'=z}^{t-1} \Pr[X'=t'] \cdot (1-m)^{t-t'-1}\, m \\
	&= \sum_{t'=z}^{t-1} \left[m^z\cdot (1-m)^{t'-z} \cdot 
	\binom {t'-1}{z-1} \right] \cdot (1-m)^{t-t'-1}\, m \\
	&= m^{z+1}\cdot(1-m)^{t-(z+1)}\cdot \sum_{t'=z}^{t-1} \binom {t'-1}{z-1} \\
 	&= m^{z+1}\cdot(1-m)^{t-(z+1)}\cdot \binom{t-1}{z} \; ,
\end{align*}
where for the second equation we have applied the induction hypothesis, and for the last one we have applied the property of binomial coefficients: $\sum_{j=k}^n \binom{j}{k} = \binom{n+1}{k+1}$.
This completes the proof.
\end{proof}

We now study the case where all geometric random variables $X_i$'s have
	distinct parameters $m_i$.
Before proving Lemma~\ref{lem:alldistinct}, we first show the following
	technical lemma.

\begin{lemma}\label{lemma:z+1}
Suppose that $m_i \neq m_j$ whenever $i \neq j$.
Then,
\begin{align}
\sum_{i=1}^\ell  \frac{(1-m_i)^{\ell-1}}{(m_i-m_{\ell+1})\cdot \prod_{j=1,j\neq i}^\ell (m_j - m_i)} = 
	\frac{(1-m_{\ell+1})^{\ell-1}}{\prod_{j=1}^\ell ( m_j - m_{\ell+1})} . \label{eqn:lemma}
\end{align}
\end{lemma}
\begin{proof}
Let $t_i\overset{\text{def}}{=}1-m_i$, then Formula~\eqref{eqn:lemma} 
	(with $\ell$ switched to $z$) becomes
\begin{align*}
\sum_{i=1}^z \frac{t_i^{z-1}}{(t_i-t_{z+1})\cdot
\Pi_{j=1,j\neq i}^z (t_j-t_i)}
=\frac{t_{z+1}^{z-1}}{\Pi_{j=1}^z (t_j-t_{z+1})} .
\end{align*}

Equivalently, we need to show the following holds:
\begin{align}
\sum_{i=1}^z \frac{t_i^{z-1}\cdot
\Pi_{j=1,j\neq i}^z (t_j-t_{z+1})}
{t_{z+1}^{z-1}
\cdot
\Pi_{j=1,j\neq i}^z (t_j-t_i)}
=1.\label{show}
\end{align}

For the base case of $z=1$, it is clear that (\ref{show}) holds. 
Suppose Formula~\eqref{show} holds for some $z > 1$, now we 
show that it also holds for the case of $z+1$.
%% First, we investigate the first $z$ items in the summation.
We first conduct the following manipulation:
\begin{align*}
&\frac{t_i^{z-1}\cdot
\Pi_{j=1,j\neq i}^{z} (t_j-t_{z+2})}
{t_{z+2}^{z-1}
\cdot
\Pi_{j=1,j\neq i}^{z} (t_j-t_i)}
-
\frac{t_i^{z}\cdot
\Pi_{j=1,j\neq i}^{z+1} (t_j-t_{z+2})}
{t_{z+2}^{z}
\cdot
\Pi_{j=1,j\neq i}^{z+1} (t_j-t_i)}
\\
=&
\frac{t_i^{z-1}\cdot
\Pi_{j=1,j\neq i}^{z} (t_j-t_{z+2})}
{t_{z+2}^{z-1}
\cdot
\Pi_{j=1,j\neq i}^{z} (t_j-t_i)}
\cdot 
\left(1-\frac{t_i\cdot (t_{z+1}-t_{z+2})}{t_{z+2}\cdot (t_{z+1}-t_i)}\right)\nonumber\\
=&
\frac{t_i^{z-1}\cdot
\Pi_{j=1,j\neq i}^{z} (t_j-t_{z+2})}
{t_{z+2}^{z-1}
\cdot
\Pi_{j=1,j\neq i}^{z} (t_j-t_i)}
\cdot
\frac{t_{z+2}t_{z+1}-t_{z+2}t_i-t_it_{z+1}+t_it_{z+2}}{t_{z+2}\cdot (t_{z+1}-t_i)}
\\
=&
\frac{t_i^{z-1}\cdot
\Pi_{j=1,j\neq i}^{z} (t_j-t_{z+2})}
{t_{z+2}^{z-1}
\cdot
\Pi_{j=1,j\neq i}^{z} (t_j-t_i)}
\cdot
\frac{t_{z+1}\cdot (t_{z+2}-t_i)}{t_{z+2}\cdot (t_{z+1}-t_i)}
\\
=&
\frac{t_i^{z-1}\cdot
\Pi_{j=1,j\neq i}^{z} (t_j-t_{z+2})}
{t_{z+2}^{z}
\cdot
\Pi_{j=1,j\neq i}^{z} (t_j-t_i)}
\cdot
\frac{t_{z+1}^z\cdot (t_i-t_{z+2})}{t_{z+1}^{z-1}\cdot(t_i-t_{z+1})}
\\
=&
\frac{t_{z+1}^z\cdot
\Pi_{j=1}^{z} (t_j-t_{z+2})}
{t_{z+2}^{z}
\cdot(t_i-t_{z+1})
}
\cdot
\frac{t_i^{z-1}}{t_{z+1}^{z-1}\cdot\Pi_{j=1,j\neq i}^{z} (t_j-t_i)}
\\
=&
\frac{t_{z+1}^z\cdot \Pi_{j=1}^{z}(t_j-t_{z+2})}{t_{z+2}^z\cdot \Pi_{j=1}^z(t_j-t_{z+1})}
\cdot
\frac{t_i^{z-1}\cdot \Pi_{j=1,j\neq i}^z(t_j-t_{z+1})}{t_{z+1}^{z-1}\cdot 
\Pi_{j=1,j\neq i}^z (t_j-t_i)}.
\end{align*}

With the above, we have
\begin{align}
&\sum_{i=1}^z\frac{t_i^{z-1}\cdot
\Pi_{j=1,j\neq i}^{z} (t_j-t_{z+2})}
{t_{z+2}^{z-1}
\cdot
\Pi_{j=1,j\neq i}^{z} (t_j-t_i)}
-
\sum_{i=1}^{z+1}\frac{t_i^{z}\cdot
\Pi_{j=1,j\neq i}^{z+1} (t_j-t_{z+2})}
{t_{z+2}^{z}
\cdot
\Pi_{j=1,j\neq i}^{z+1} (t_j-t_i)}
\label{substitute}
\\
=&
\sum_{i=1}^z\frac{t_i^{z-1}\cdot
\Pi_{j=1,j\neq i}^{z} (t_j-t_{z+2})}
{t_{z+2}^{z-1}
\cdot
\Pi_{j=1,j\neq i}^{z} (t_j-t_i)}
-
\sum_{i=1}^{z}\frac{t_i^{z}\cdot
\Pi_{j=1,j\neq i}^{z+1} (t_j-t_{z+2})}
{t_{z+2}^{z}
\cdot
\Pi_{j=1,j\neq i}^{z+1} (t_j-t_i)}
-
\frac{t_{z+1}^z\cdot \Pi_{j=1}^{z}(t_j-t_{z+2})}
{t_{z+2}^z\cdot \Pi_{j=1}^z(t_j-t_{z+1})} \nonumber \\
&=
\frac{t_{z+1}^z\cdot \Pi_{j=1}^{z}(t_j-t_{z+2})}{t_{z+2}^z\cdot \Pi_{j=1}^z(t_j-t_{z+1})}
\cdot
\sum_{i=1}^z
\frac{t_i^{z-1}\cdot \Pi_{j=1,j\neq i}^z(t_j-t_{z+1})}{t_{z+1}^{z-1}\cdot
\Pi_{j=1,j\neq i}^z (t_j-t_i)}
-
\frac{t_{z+1}^z\cdot \Pi_{j=1}^{z}(t_j-t_{z+2})}
{t_{z+2}^z\cdot \Pi_{j=1}^z(t_j-t_{z+1})}\nonumber\\
&=
\frac{t_{z+1}^z\cdot \Pi_{j=1}^{z}(t_j-t_{z+2})}{t_{z+2}^z\cdot \Pi_{j=1}^z(t_j-t_{z+1})}
-
\frac{t_{z+1}^z\cdot \Pi_{j=1}^{z}(t_j-t_{z+2})}{t_{z+2}^z\cdot \Pi_{j=1}^z(t_j-t_{z+1})}
=0, \label{basedonassumption}
\end{align}
where \eqref{basedonassumption} is obtained by applying the induction hypothesis.
Also notice that in (\ref{substitute}),
$$
\sum_{i=1}^z\frac{t_i^{z-1}\cdot
\Pi_{j=1,j\neq i}^{z} (t_j-t_{z+2})}
{t_{z+2}^{z-1}
\cdot
\Pi_{j=1,j\neq i}^{z} (t_j-t_i)} = 1,
$$
based again on the induction hypothesis (with $t_{z+1}$ replaced by $t_{z+2}$).

Hence, we have 
\[1-\sum_{i=1}^{z+1}\frac{t_i^{z}\cdot
\Pi_{j=1,j\neq i}^{z+1} (t_j-t_{z+2})}
{t_{z+2}^{z}
\cdot
\Pi_{j=1,j\neq i}^{z+1} (t_j-t_i)}=0
\]
This implies that equation (\ref{show}) holds for all $z\geq 1$, which completes the proof.
\end{proof}

We are now ready to prove Lemma~\ref{lem:alldistinct}.

\paragraph{Lemma~\ref{lem:alldistinct} (re-stated).}
{\lemalldistinct}

\begin{proof}
We prove the lemma by an induction on the path length $\ell$.
For the base case when $\ell=1$, $X=X_1$, and thus
	$\Pr[X=t] = (1-m_1)^{t-1} m_1$ for all $t\ge 1$.
For the induction step, 
	suppose that the lemma is true for $\ell = z$ for some $z\in \mathbb{Z}_+$.
Let $X'=\sum_{i=1}^{z} X_i$ and
	$X = \sum_{i=1}^{z+1} X_i$.
By the induction hypothesis, we have that for all $t\ge z$, 
$$
\Pr[X'=t] = \left(\prod_{i=1}^z m_i \right) \cdot \sum_{i=1}^z \frac{(1-m_i)^{t-1}}{\prod_{j = 1, j \neq i}^z (m_j-m_i)}\,.
$$

Now consider the case of $\ell = z+1$. It is clear that for $t < z+1$,
	$\Pr[X=t] =0$.
For $t \ge z+1$, the event $\{X=t\}$ is the union of mutually exclusive
	events of $\{X'=t', X_{z+1}= t - t'\}$ for all $t'=z,z+1,\ldots, t-1$.
Thus we have
\begin{align*}
\Pr[X=t] &= \sum_{t'=z}^{t-1} \Pr[X'=t']\cdot (1-m_{z+1})^{t-t'-1}\cdot m_{z+1}\\
		     &= \left( \prod_{i=1}^{z+1}m_i\right) \cdot  \sum_{t'=z}^{t-1}\,\sum_{i=1}^z \frac{(1-m_i)^{t'-1}}{\prod_{j = 1, j \neq i}^z (m_j-m_i)} \cdot (1-m_{z+1})^{t-t'-1}\,,\\
\end{align*}
where we have applied the induction hypothesis
	for the last equality.

Next, we switch the order of the two summations and get
\begin{align*}
\Pr[X=t] 
	&= \left(\prod_{i=1}^{z+1}m_i\right) \cdot \sum_{i=1}^z\,\sum_{t'=z}^{t-1} \frac{(1-m_i)^{t'-1}}{\prod_{j = 1, j \neq i}^z (m_j-m_i)} \cdot (1-m_{z+1})^{t-t'-1} \\
	&= \left(\prod_{i=1}^{z+1}m_i\right)\cdot \left[ \sum_{i=1}^z \frac{(1-m_{z+1})^{t-2}}{\prod_{j=1,j\neq i}^z (m_j-m_i)} \cdot \sum_{t'=z}^{t-1} \left( \frac{1-m_i}{1-m_{z+1}} \right)^{t'-1} \right]\\
	&= \left(\prod_{i=1}^{z+1}m_i\right)\cdot \left[ \sum_{i=1}^z \frac{(1-m_{z+1})^{t-2}}{\prod_{j=1,j\neq i}^z (m_j-m_i)} \cdot \frac{ \left(\frac{1-m_i}{1-m_{z+1}}\right)^{z-1}  - \left( \frac{1-m_i}{1-m_{z+1}} \right)^{t-1} }{1-\left(\frac{1-m_i}{1-m_{z+1}}\right)} \right] \\
\end{align*}

\begin{align}
\qquad\qquad\qquad &= \left(\prod_{i=1}^{z+1}m_i \right)\cdot \sum_{i=1}^z \frac{(1-m_i)^{z-1}(1-m_{z+1})^{t-z} - (1-m_i)^{t-1}}{(m_i-m_{z+1})\cdot \prod_{j=1,j\neq i}^{z} (m_j - m_i)} \nonumber\\
\qquad\qquad\qquad &= \left(\prod_{i=1}^{z+1}m_i \right)\cdot \left[ \sum_{i=1}^z \frac{(1-m_i)^{t-1}}{\prod_{j=1,j\neq i}^{z+1} (m_j - m_i)} +  \sum_{i=1}^z  \frac{(1-m_i)^{z-1}(1-m_{z+1})^{t-z}}{(m_i-m_{z+1})\cdot \prod_{j=1,j\neq i}^{z} (m_j - m_i)} \right] \label{eq:lastline}
\end{align}

For the second summation  term in \eqref{eq:lastline},
	we need to prove the following equation:
\begin{equation}\label{eq:z+1}
\sum_{i=1}^z  \frac{(1-m_i)^{z-1} (1-m_{z+1})^{t-z}}{(m_i-m_{z+1})\cdot \prod_{j=1,j\neq i}^{z} (m_j - m_i)} = 
	\frac{(1-m_{z+1})^{t-1}}{\prod_{j=1}^{z}( m_j - m_{z+1})} .
\end{equation}

This can be done by dividing both sides of \eqref{eq:z+1} by $(1-m_{z+1})^{t-z}$ and applying Lemma~\ref{lemma:z+1}.
Then, we substitute \eqref{eq:z+1} back into \eqref{eq:lastline} and get
$$
\Pr[X=t] = \left(\prod_{i=1}^{z+1}m_i \right) \cdot \sum_{i=1}^{z+1} \frac{(1-m_i)^{t-1}}{\prod_{j=1,j\neq i}^{z+1} (m_j - m_i)} .
$$

Thus, the lemma holds for $\ell = z+1$, and this completes the proof.
\end{proof}

\section{Time-Critical Influence Maximization with Login Events}\label{sec:login}
In online social networks, users usually will not stay online interacting with friends all the time.
For example, on Twitter or Facebook, users may respond to posts generated by friends perhaps hours or even days ago.
This time-delayed behavior can be modeled using the probabilities of logging into the social systems.
Let $\ell_u \in [0,1]$ denote the {\em login probability} of user $u\in V$.
In any time step $t$, each user $u\in U$ independently logs into the social networking site with probability $\ell_u$, or stays offline with probability $1-\ell_u$.
Next, we describe how to incorporate the login events into both IC and LT models to reflect time-delayed influence diffusion processes.

\paragraph{Independent Cascade with Login Events (IC-L)}
In the \emph{IC-L} model, we start with a seed set $S$ at time step $0$.
In every time step $t\ge 0$, a node $u$ has an independent online probability $\ell_u$.
For $s \in S$, it becomes active when logging into the system for the first time.

If a node $u$ becomes active in step $t$, then $u$ will have a single chance to influence each of its inactive neighbor $v$ at step $t' \geq t+1$, when $v$ logs in  
	for the first time after $t$. This attempt has a success probability of $p(u,v)$.
If the attempt succeeds, $v$ will become active at $t'$; 
%\weic{The old version said $v$ activates at time $t'+1$. I don't know the exact reason. I think
%it is because we do not want the case that $u$ logs in at time $t'$, gets activated, and
%then $v$ logs in at the same time $t'$ and also gets activated by $u$. But if we already specify that
%$u$ has to pass influence to future logins, $t'>t$ above, so I think we do not need $t'+1$.
%All these timing may affect our later dynamic programming, so need further checking.}
otherwise, $u$ cannot attempt to influence $v$ even if $v$ logs in again in the future.
The diffusion process terminates either naturally, i.e., when no more nodes can be activated, or by a specific deadline, i.e., the end of time step $\tau$.

%\weil{Dr. Chen, I agree with your change highlighted in blue.  We can assume that friends do not see concurrent activities at the same time step. That is, even $u$ and $v$ both online at $t$, and $u$ becomes active at $t$, but $v$ won't see it until the next time he logs in, which will be $\geq t+1$.
%This is in accordance with the definition of IC, and only in this way, IC-L subsumes IC as a special case when all $\ell_u = 1$.}

It can be easily shown that the influence maximization problem is also \NPhard{} under the IC-L model, as we can restrict all login probabilities to be $1$, making IC-L equivalent to IC.

%\begin{theorem}
%Influence maximization under the IC-L model is \NPhard{}.
%\end{theorem}
%
%\weil{We probably do not need a standalone theorem statement for this result, as it's quite straightforward.  Similarly for Theorem 9 below.}
%
%\weic{Agree. Just stating them in text is fine.}

\paragraph{Linear Thresholds with Login Events (LT-L)}
Similar to the traditional LT model, in the LT-L model, each node $v\in V$ chooses a threshold $\theta_v$ uniformly at random from $[0,1]$, and it is influenced by its neighbors $u$ based on edge weight $b(u,v)$.

The diffusion dynamics unfold as follows.
First, a seed set $S$ is targeted at time step $0$, and for any $s\in S$, $s$ becomes active
	when it first logs into the system.
In every subsequent time step $t\ge 0$, each node $v$ logs in with probability $\ell_v$.
If a node $u$ becomes active in step $t$, we say $u$ is an \emph{effective} active in-neighbor of its currently inactive out-neighbor $v$ by time $t$, meaning that $u$ can pass its influence $b(u,v)$ to $v$.
An inactive $v$ would get activated at step $t'$ if $v$ logs in and the total weight of its effective active neighbors by time $t'-1$ is at least $\theta_v$:
$$\sum_{u \in EA_v^t} b(u,v) \geq \theta_v.$$
The propagation of influence stops either naturally or by the end of deadline $\tau$.
Similarly, influence maximization under the LT-L model is also \NPhard{}, due to that
	the LT-L model subsumes LT when all login probabilities are $1$.
%
%\begin{theorem}
%	Influence maximization under the LT-L model is \NPhard{}.
%\end{theorem}

\subsection{Submodularity and Approximation Guarantees}
Kempe et al.~\cite{kempe03} shows that both IC and LT are special cases of the \emph{Triggering Set} (TS) model, in which each node $v$ randomly picks a \emph{Triggering Set} $T(v)$, a subset of its in-neighbors $\Nin(v)$, according to some distribution over subsets of $\Nin(v)$.
An inactive $v$ becomes active at time step $t$ if there exists some $u \in T(v)$ that are active by $t-1$.

Combining the TS model with deadline constraint and login probability, we propose the \emph{TS-L model}. In this model, an inactive $v$ becomes active at time step $t$ if it logs in at $t$ and notices {\sl for the first time} that there exists $u \in T_v$ such that $u$ is already active (i.e., before $t$). 
Note that TS-L generalizes IC-L and LT-L:
\begin{itemize}
\item For the IC-L model, the triggering set $T(v)$ of each user $v$ is to include every in-neighbor $u\in \Nin(v)$ independently with probability $p(u,v)$.
\item For the LT-L model, the triggering set $T(v)$ consists of at most one in-neighbor of $v$: a particular $u\in \Nin(v)$ is chosen with probability $b(u,v)$. And with probability $1-\sum_{u \in \Nin(v)} b(u,v)$, $T(v) = \emptyset$.
\end{itemize}
In all three models, the login events happen independently for all users at all time steps.
Thus, TS-L indeed includes IC-L and LT-T as special cases.

%\weil{Previously Ning had a lemma for this, but the proof seems wordy and not to the point.  I re-wrote the above.  I think it suffices to just argue that how triggering set is generated in IC and LT. The login events are the same in all three models, so the generalization argument holds. }
%
%\weic{Agree.}

Next, we show that the influence spread function is submodular for the TS-L model, which then implies that submodularity also holds for IC-L and LT-L.
%\emph{first time} means that $v$ did not log in after $t_1$ until $t$.

%\begin{lemma}
%	TS-L generalizes IC-L and LT-L.
%\end{lemma}
%\begin{proof}
%\textcolor{red}{Need to re-write this proof}
%The Triggering Set model with deadline constraint and random meeting events (TS-L) still generalizes IC-L and LT-L because the generality is in terms of the way we generate possible worlds (deterministic live-edge graphs) out of the original social influence graph, and the propagation of influence becomes a reachability problem, which gives us the submodularity result.
%On the other hand, the deadline constraint controls when the diffusion ends but it does not affect the dynamics of the diffusion process.
%In addition, the effects of random meeting events can be tackled by ``re-defining'' reachability on live-edge graphs, which is exactly what we have done in the proof of Theorem~\ref{theorem:submod}, and those meeting events essentially just ``slow down'' the diffusion of influence, but does not affect how it proceeds, either.
%Thus, by Theorem~\ref{theorem:submod} we can see both IC-L and LT-L, being special cases of TS-L, have a monotone submodular influence function, which allows the greedy algorithm to achieve a $1-1/e$ approximation guarantee.
%\end{proof}

\begin{theorem}\label{theorem:submod}
	The influence function $\sigma(\cdot)$ is monotone and submodular under the TS-L model.
\end{theorem}

%\weil{Is the bold part highlighted by Ning originally?  I think its not entirely sound.
%First, ``collective wating time'' is not defined properly.
%Second, I think the proof borrows ideas from the SM proof for IC-M model, which also has a similar notion of collective waiting time. 
%However, this is a significant difference: in IC-M, seeds become active at $t=0$, but here, seeds become active at the first time step they log in, not necessarily $0$, so the waiting time $\leq \tau$ is not enough to guarantee that the end node of this path becomes active by time step $\tau$.}

\begin{proof}
First, we fix a set $X_T$ of outcomes of triggering set selections for all nodes. In this fixed $X_T$, if $u \in T(v)$, then we declare $(u,v)$ to be \emph{live}; otherwise we declare it \emph{blocked}.
Next, for each $u$ at each time step $t \in [0,\tau]$, we independently flip a coin with bias $\ell_u$ to determine whether $u$ will log in in $t$.
Eventually, we obtain a sequence of log-in events for $u$.
Let $X_L$ be the set of such sequences of all pairs.

Any fixed $X_L$, on top of a fixed $X_T$, forms a possible world $X$ where influence propagates deterministically. %, which is essentially a deterministic graph of certain existence probability.
First, consider a live-edge path $\P = \{u_0, u_1, \ldots, u_z\}$ in possible world $X$, where $z$ is the length of $\P$.
Suppose, without loss of generality, that $u_0 \in S$ and $u_i\not\in S$, for all $1\leq i\leq z$.
Note that if none of the nodes in $\P$ is a seed, such a path can be ignored in influence propagation.

Let $t(\P,u_0)$ be the time step at which $u_0$ becomes active (since $u_0\in S$, technically the activation time of $u_0$ does not depend on $\P$, and we use this notation simply for technical convenience).
For all $1 \leq i \leq z$, let $t(\P, u_i)$ be the first time step at which $u_i$ logs in after the activation of $u_{i-1}$, its predecessor on $\P$.
Hence, by model definition, $u_i$ will become active at $t(\P,u_i)$ \textsl{if we ignore other paths along which influence may propagate to $u_i$}.
Clearly, if $t(\P, u_z) \leq \tau$, then the end node $z$ becomes active by the deadline and should count toward influence spread.

%Clearly, $CWT(\P)$ is the total number of time steps required for the propagation to reach the end of this path, $u_z$, all the way from $u_0$.
%In some sense, this can be seen as a ``weighted'' length of path $\P$.
%

We now define the notion of \emph{reachability} in a possible world $X$, which is slightly different from the traditional reachability in graphs due to log-in events.
We say $v$ is \emph{reachable from seed set $S$} iff (1) there exists a live-edge path $\P_{S,v}$ from some $s \in S$ to $v$, and %(2) the \emph{collective waiting time} of $\P_{S,v}$, determined by the influence propagation and the sequences of log-in event of all $\ell_u \in \P$, is no longer than $\tau$. 
(2) $t(\P_{S,v},v) \leq \tau$.
Note that if there are multiple live-edge paths from $S$ to $v$ in this possible world, we take the shortest one in terms of the activation time of $v$.

%\weic{I prefer not using boldface in the text, so changed above.}

Let $\sigma_X(S)$ be the number of nodes reachable from $S$ by the reachability definition above.
Now consider two seed set $S_1$ and $S_2 \supseteq S_1$, and a node $x \in V \setminus S_2$:
$\sigma_X(\cdot)$ is clearly monotone as if $u$ can be reached by $S_1$, then the origin of the live-edge path to $u \in S_1$ must also belong to $S_2$, thus we have $\sigma_X(S_1) \leq \sigma_X(S_2)$.
For submodularity, consider a node $u$ is reachable from $S_2 \cup \{x\}$ but not $S_2$, then we conclude immediately that (1) $u$ is not reachable from $S_1$ either, and (2) the origin of the live-path to $u$ must be $x$.
Hence, $u$ is reachable from $S_1 \cup \{x\}$ but not $S_1$, and we have $\sigma_X(S_1\cup\{x\})-\sigma_X(S_1) \geq \sigma_X(S_2\cup\{x\})-\sigma_X(S_2)$.

The influence function $\sigma(\cdot)$ can be written as a nonnegative linear combination of $\sigma_X(\cdot)$ functions: for any $S \subseteq V$
$$\sigma(S) = \sum_{X} \Pr[X]\cdot\sigma_X(S)$$
where $X$ is any combination of $X_T$ and $X_M$, and $\Pr[X] = \Pr[X_T]\cdot\Pr[X_M]$.
Hence, $\sigma(\cdot)$ is also monotone and submodular, which completes the proof.
\end{proof}

\subsection{Influence Computation under IC-L Model and LT-L Model}

We now consider the problem of computing influence
	spread under these two models.
Since IC-L (LT-L) subsumes the classical IC (resp.\ LT) model
	as a special case (by setting all login probabilities to be $1$),
	computing the exact value of $\sigma(S)$ for any seed
	set $S\subseteq V$ remains \SPhard~\cite{ChenWW10,
	ChenYZ10}.
Therefore, we focus on the computation of influence spread in local
	structures, i.e., Maximum Influence Arborescence (MIA) for IC-L
	and Local Directed Acyclic Graphs (LDAG)  for LT-L.
In what follows, we give formulas for computing the activation probability
	of a node, given its local influence structure.
The construction of such local influence structures can be done
	is a similar fashion to MIA-M (Section~\ref{sec:algo}) and LDAG-M (Section~\ref{sec:ldag}),
	and we omit details here.

\subsubsection{IC-L Model: Partial Results for Computations in an In-Arborescence}
For the IC-L model, since the login probability is the property of node instead of edge, the events of users' activated by different in-neighbors are no longer independent.
For example, suppose node $v$ has two in-neighbors $u_1$ and $u_2$, and suppose that
	both $u_1$ and $u_2$ are seeds.
Then the event that $u_1$ activates $v$ by time $\tau$ is the joint event of 
	(a) $u_1$ logged in at some time $t < \tau$ and get activated;
	(b) $v$ logged in at some time $t'$ with $t < t' \le \tau$; and 
	(c) $u_1$ successfully influenced $v$ at time $t'$ when $v$ logged in.
Similarly we have the event that $u_2$ activates $v$ by time $\tau$ as the three
	parallel joint events.
When comparing these events, we can see that (a) and (c) for $u_1$ and $u_2$ are independent, 
	but not (b), since they are both for the login event of $v$.
This is different from the meeting probability model IC-M we provide before, in which case
	we would replace (b) above with $u_1$ meeting $v$ at time $t'$, which is indeed independent
	of $u_2$ meeting $v$ at time $t'$.
As the result, we need a new dynamic programming method to compute the activation
	probabilities and influence spread, even in local influence regions.

Consider an in-arborescence $H = (V_H, E_H)$ and a node $v\in V_H$.
Let $\mathcal{A}^{in}(v,t) = \{u_1, u_2, \ldots, u_m\} \subseteq \Nin_H(v)$ denote the set of $v$'s active in-neighbors (in $H$) at the beginning of time step $t$ and let $m = |\mathcal{A}^{in}(v,t)|$.
We order these nodes by their activation time (ascending order, with ties broken arbitrarily). %: $\langle u_1,u_2,\ldots, u_m \rangle$, such that  $t_1\leq t_2\leq \ldots \leq t_m$ (ties broken arbitrarily).
Also, let $p_i$ be the influence probability $p(u_i,v)$. % and $t_i$ is the activation time step of $u_i$.
Let $T_{[1,m]} = \langle t_1,t_2, \ldots, t_m \rangle$ be the ordered sequence of activation time steps of nodes in  $\mathcal{A}^{in}(v,t)$.
Let $P_{[1,m]} = \langle p_1, p_2, \ldots, p_m  \rangle$ be the corresponding sequence of influence probabilities.

\begin{theorem}\label{thm:ap-ICL}
Consider a node $v\in V_H \setminus S$.
Given $T_{[1,m]}$ and $P_{[1,m]}$ (that correspond to a particular sequence of active in-neighbors),  $v$'s {\em activation probability} $ap(v,t,T_{[1,m]},P_{[1,m]})$, $\forall t \geq t_m$, can be computed as follows:
\begin{multline}\label{eqn:ap-ICL}
    ap(v,t,T_{[1,m]},P_{[1,m]}) = (1-(1-\ell_v)^{t_2-t_1}) \cdot(1-p_1)\cdot ap(v,t, T_{[2,m]}, P_{[2,m]}) 
            +(1-\ell_v)^{t_2-t_1}\cdot ap(v,t, T_{[2,m]}, P'_{[2,m]}),
\end{multline}
where $m\geq 3$ and $P'_{[2,m]} = 1-(1-p_1)(1-p_2) \oplus P_{[3,m]}$, with $\oplus$ denoting the operation of sequence concatenation. 
There are two base cases.
\begin{itemize}
\item If $m=1$, we have $ap(v,t, \langle t_1 \rangle, \langle p_1 \rangle)=(1-\ell_v)^{t-t_1-1}\cdot \ell_v\cdot p_1$.
\item If $m=2$, apply Equation~\eqref{eqn:ap-ICL} with $P_{[3,m]}$ being an empty sequence.
\end{itemize}
\end{theorem}

%\weil{made some changes to the proof, although it was generally sound.  Pls check.}

\begin{proof}
It is straightforward to verify the correctness of the two bases cases.
For $m\geq 3$, two cases arise as sub-problems:

$(i)$.
If $v$ has logged in at least once during this period with probability $1 - (1-\ell_v)^{t_2-t_1}$,
	then in order for $v$ to be active at time $t$, $u_1$ must have tried but failed to
	activate $v$, which happens with probability $1-p_1$.
Hence, $u_1$ can be disregarded and we only need to consider the subproblem on
	the remaining in-neighbor sequence $\{u_2, u_3, \ldots, u_m\}$.
This gives the first summation term in Equation~\eqref{eqn:ap-ICL}.

$(ii)$.
If $v$ never logged in during $(t_1,t_2]$,
%$u_1$'s influence should be considered together with $u_2$ from the time step $t_2$. 
	then, when it first logs in at some time step $t' \geq t_2+1$, $u_1$ and
	$u_2$ each has a single chance to independently activate $v$.
Their ``collective'' influence probability is $1-(1-p_1)(1-p_2)$.
This leads to the second summation term in in Equation~\eqref{eqn:ap-ICL}.

Since $(i)$ and $(ii)$ are mutually exclusive, the recursive formula
	correctly computes the activation probability.
Note that if there are $k\ge 2$ in-neighbors whose activation time steps are the same, we can ``collapse'' them into a single node, influence probability $1-\prod_{i=1}^{k}(1-p_i) $, similar to case $(ii)$ above.
\end{proof}

Computing Equation~\eqref{eqn:ap-ICL} takes $O(m)$ time, but it only works for a fixed sequence of active in-neighbors.
Given a seed set $S$, there may well be multiple possible sequences of in-neighbors to be considered, and the final activation probability of $v$ at time step $t$ is a linear combination of $ap(v,t)$ w.r.t.\ all applicable sequences.
Thus, Theorem~\ref{thm:ap-ICL} should only be regarded
	as a partial result, and further analysis is needed to obtain the final dynamic programming formula for
	computing activation probability in arborescences for the IC-L model.

Also, note that Theorem~\ref{thm:ap-ICL} is applicable to all non-seed nodes.
For seeds, dynamic programming is not required:
Consider any $s\in S$, the probability that $s$ becomes active at time step $t$
	is simply $(1-\ell_s)^t \cdot \ell_s$.

\subsubsection{LT-L Model: Computations in Directed Acyclic Graphs}
For the LT-L model,  %since does not need in-neighbors to pass their influences independently,
activation probabilities can be similarly computed as in the LT-M model (Theorem~\ref{thm:ltap}).
\begin{theorem} %\label{theorem:ap for lt}
Consider any directed acyclic graph $H=(V_H, E_H)$, and a node $v\in V_H\setminus S$.
For any time step $t\in[1,\tau]$,
	the {\em activation probability} of $v$ at $t$ is
\begin{equation} %\label{eq:ltap}
ap(v,t) = \sum_{u\in\Nin_H(v)}b(u,v)\sum_{t'=0}^{t-1} ap(u,t')\cdot \ell_v\cdot (1-\ell_v)^{t-t'-1},
\end{equation}
where $\Nin_H(v)$ is the set of in-neighbors of $v$ in $H$.
\end{theorem}

\begin{proof}[Proof (Sketch)]
The arguments in the proof of Theorem~\ref{thm:ltap} also apply to the LT-L model.
The main observation is that in the live-edge graph model equivalent to the LT model, 
	each node $v$ only randomly select one in-neighbor $u$ and make edge $(u,v)$ live.
With only one in-coming live edge $(u,v)$, whether we consider the event of $v$ meeting 
	$u$ after $u$ is activated 
	as in the LT-M model or the event $v$ logging in after $u$ is activated at in
	the LT-L model is essentially the same.
\end{proof}

Again, for seeds, the probability that a seed $s\in S$ becomes active at time step $t$
	is simply $(1-\ell_s)^t \cdot \ell_s$.

%\weil{It appears that as though LT-L model could be ``reduced to'' LT-M model, by setting $m(u,v)$, the probability $u$ meets $v$, to be the login probability of $\ell_v$.
%However, the difference in seed behaviours will break the equivalence.
%Fortunately, the proof for Theorem~\ref{thm:ltap} has nothing to do with seeds, so it should carry over to this model.
%Pls double-check.}
%
%\weic{I added one more sentence for an explanation.}

%\section{Time-Critical Influence Maximization in General Threshold Model}\label{sec:general}
%\input{general}

\end{document}